\def\confversion{0}
\def\ifconf{\ifnum\confversion=1}
\def\ifnotconf{\ifnum\confversion=0}

\documentclass[11 pt]{article}
\def\showauthornotes{1}
\def\showkeys{0}
\def\showdraftbox{0}
\def\confversion{0}
\def\widemargin{1}
\def\ipadcompile{0}
\usepackage{xspace,enumerate}
\usepackage{amsmath,amssymb}
\usepackage{amsthm}
\ifnum\confversion=0
	\usepackage[toc,page]{appendix}
\fi
\usepackage[T1]{fontenc}
\usepackage[utf8]{inputenc}

\usepackage{thmtools}
\usepackage{thm-restate}
\usepackage{color,graphicx}
\usepackage{boxedminipage}
\usepackage{makecell}
\usepackage{tabularx,multirow}
\usepackage{csquotes}
\usepackage{mdframed}

\ifnum\widemargin=0
\usepackage[top=1in, bottom=1in, left=1in, right=1in]{geometry}
\else
\usepackage[top=1in, bottom=1in, left=1.25in, right=1.25in]{geometry}
\fi

\ifnum\showkeys=1
\usepackage[color]{showkeys}
\fi

\definecolor{darkred}{rgb}{0.5,0,0}
\definecolor{darkgreen}{rgb}{0,0.35,0}
\definecolor{darkblue}{rgb}{0,0,0.55}

\usepackage[dvipsnames]{xcolor}

\usepackage[pdfstartview=FitH,pdfpagemode=UseNone,colorlinks,linkcolor=NavyBlue,filecolor=blue,citecolor=OliveGreen,urlcolor=NavyBlue,pagebackref]{hyperref}

\usepackage{enumitem}
\usepackage{comment}
\usepackage[color=asparagus]{todonotes}

\usepackage{tikz-cd}

\usepackage[capitalise,nameinlink]{cleveref}

\usepackage{microtype}

\usepackage{mathtools,dsfont, bbm}

\usepackage{palatino}
\usepackage[scaled=.95]{helvet}

\ifnum\ipadcompile=1
\usepackage{eulervm}
\else
\usepackage{eulerpx}
\fi

\DeclareFontEncoding{LGR}{}{}
\DeclareSymbolFont{sfitgreek}{LGR}{cmss}{m}{it}
\SetSymbolFont{sfitgreek}{bold}{LGR}{cmss}{bx}{it}
\DeclareMathSymbol{\sfpi}{\mathord}{sfitgreek}{`p}

\usepackage{tcolorbox}

\DeclareMathAlphabet{\mathpazocal}{OMS}{zplm}{m}{n}
\DeclareRobustCommand*{\mathcal}[1]{\mathpazocal{#1}}
\ifnum\showauthornotes=1
\newcommand{\Authornote}[3]{{\sf\small\color{#3}{[#1: #2]}}}
\newcommand{\Authorcomment}[2]{{\sf \small\color{gray}{[#1: #2]}}}
\newcommand{\Authorfnote}[2]{\footnote{\color{red}{#1: #2}}}
\else
\newcommand{\Authornote}[3]{}
\newcommand{\Authorcomment}[2]{}
\newcommand{\Authorfnote}[2]{}
\fi

\ifnum\showdraftbox=1
\newcommand{\draftbox}{\begin{center}
  \fbox{%
    \begin{minipage}{2in}%
      \begin{center}%
        \begin{Large}%
          \textsc{Working Draft}%
        \end{Large}\\
        Please do not distribute%
      \end{center}%
    \end{minipage}%
  }%
\end{center}
\vspace{0.2cm}}
\else
\newcommand{\draftbox}{}
\fi

\declaretheorem[numberwithin=section]{theorem}
\declaretheorem[sibling=theorem]{lemma}
\declaretheorem[sibling=theorem]{claim}
\declaretheorem[sibling=theorem]{proposition}
\declaretheorem[sibling=theorem]{fact}
\declaretheorem[sibling=theorem]{corollary}

\theoremstyle{definition}
\declaretheorem[sibling=theorem]{definition}
\declaretheorem[sibling=theorem]{remark}

\declaretheorem[sibling=theorem]{observation}
\declaretheorem[sibling=theorem]{example}

\newtheorem{algo}[theorem]{Algorithm}
\def\FullBox{\hbox{\vrule width 6pt height 6pt depth 0pt}}

\def\qed{\ifmmode\qquad\FullBox\else{\unskip\nobreak\hfil
\penalty50\hskip1em\null\nobreak\hfil\FullBox
\parfillskip=0pt\finalhyphendemerits=0\endgraf}\fi}

\def\qedsketch{\ifmmode\Box\else{\unskip\nobreak\hfil
\penalty50\hskip1em\null\nobreak\hfil$\Box$
\parfillskip=0pt\finalhyphendemerits=0\endgraf}\fi}

\DeclareMathOperator{\agr}{agr}

\DeclareMathOperator{\Ima}{Im}
\newcommand{\Hx}{\sfH_{\vec x}}
\newcommand{\tG}{\widetilde{G}} 
\newcommand{\tH}{\widetilde{H}} 
\newcommand{\trho}{\tilde{\rho}} 
\newcommand{\LieHom}{\textup{LieHom}}
\newcommand{\LiftHom}{\textup{LiftHom}}

\newcommand{\indi}{\mathds{1}}
\newcommand{\indicator}[1]{\mathds{1}_{#1}\xspace}
\newcommand\inv[1]{#1\raisebox{0.8ex}{$\scriptscriptstyle-1$}}
\newcommand{\D}{D}
\newcommand{\Inn}{\textup{Inn}}
\newcommand{\Hom}{\textup{Hom}}
\newcommand{\Aut}{\textup{Aut}}
\newcommand{\Sym}{\textup{Sym}}
\newcommand{\stab}{\textup{Stab}}

\newcommand{\SL}{\mathrm{SL}}

\newcommand{\GL}{\mathrm{GL}}

\newcommand{\triv}{\mathrm{triv}}

\newcommand{\sfT}{{\sf T}}
\newcommand{\sfG}{{\sf G}}

\newcommand{\ep}{\varepsilon}

\newcommand{\sfH}{{\sf H}}
\newcommand{\sfS}{{\sf S}}

\DeclareMathOperator{\rank}{rank}

\newcommand{\test}{{\sf Test}}

\newcommand{\vgen}[1]{\langle {#1} \rangle}

\newcommand{\cotype}[1]{{\mathrm{cotype}} ({#1})}

\newcommand{\ctuple}[2]{ {\sf Count}_{#1} ({#2})}
\newcommand{\vtest}{\indicator{f(\vec x)\in \Hx }  }

\newcommand{\tq}{\tilde{q}}
\newcommand{\fix}{\textup{Fix}}

\newcommand{\agrHom}{\agr{(f,\phi)}}
\newcommand{\tagrHom}{\widetilde{\agr}{(f,\phi)}}

\makeatletter
\newcommand{\stackalign}[1]{
	\vcenter{
		\Let@ \restore@math@cr \default@tag
		\baselineskip\fontdimen10 \scriptfont\tw@
		\advance\baselineskip\fontdimen12 \scriptfont\tw@
		\lineskip\thr@@\fontdimen8 \scriptfont\thr@@
		\lineskiplimit\lineskip
		\ialign{\hfil$\m@th\scriptstyle##$&$\m@th\scriptstyle{}##$\crcr
			#1\crcr
		}
	}
}

\let\latexcirc=\circ
\newcommand{\ccirc}{\mathbin{\mathchoice
  {\xcirc\scriptstyle}
  {\xcirc\scriptstyle}
  {\xcirc\scriptscriptstyle}
  {\xcirc\scriptscriptstyle}
}}
\newcommand{\xcirc}[1]{\vcenter{\hbox{$#1\latexcirc$}}}\let\circ\ccirc

\def\to{\rightarrow}
\def\eps{\varexpan}
\def\epsilon{\varepsilon}

\def\eps{\epsilon}

\def\phi{\varphi}
\def\cal{\mathcal}

\def\implies{\Rightarrow}

\newcommand{\ie}{i.e.,\xspace}

\newcommand{\E}{{\mathbb E}}
\newcommand{\C}{{\mathbb C}}

\newcommand{\Z}{{\mathbb Z}}

\newcommand{\F}{{\mathbb F}}

\newcommand{\cG}{\mathcal{G}}

\newcommand{\cD}{\mathcal{D}}

\DeclarePairedDelimiter\parens{\lparen}{\rparen}
\DeclarePairedDelimiter\abs{\lvert}{\rvert}
\DeclarePairedDelimiter\norm{\lVert}{\rVert}

\DeclarePairedDelimiter\braces{\lbrace}{\rbrace}
\DeclarePairedDelimiter\brackets{\lbrack}{\rbrack}
\DeclarePairedDelimiter\angles{\langle}{\rangle}
\DeclarePairedDelimiterXPP\lnorm[1]{}\lVert\rVert{_2}{#1}

\DeclareMathDelimiter{\given}
      {\mathbin}{symbols}{"6A}{largesymbols}{"0C}

\newcommand{\prob}{\mathsf{Pr}}
\newcommand{\Esymb}{\mathbb{E}}

\newcommand{\Psymb}{\mathrm{Pr}}

\newcommand{\talpha}{\tilde{\alpha}}

\DeclarePairedDelimiterXPP{\Prob}[1]
 {\prob}{\lparen}{\rparen}{}
 {\renewcommand{\given}{\;\delimsize\vert\nonscript\;\mathopen{}}#1}

\makeatletter
\def\Pr#1{%
    \ProbabilityRender{\Psymb}{#1}%
}

\def\Ex#1{%
    \ProbabilityRender{\Esymb}{#1}%
}

\def\condPE#1#2{%
	\@ifnextchar\bgroup
	{\ConditionalProbabilityRender{\widetilde{\Esymb}}{#1}{#2}}
	{\ProbabilityRender{\widetilde{\Esymb}}{#1 \given #2}}
}

\def\ConditionalProbabilityRender#1#2#3#4{
	\renderwithdist{#1}{#2}{#3 \given #4}	
}

\def\ProbabilityRender#1#2{%fancy probability command
  \@ifnextchar\bgroup%
  {\renderwithdist{#1}{#2}}
   {\singlervrender{#1}{#2}}
}
\def\singlervrender#1#2{%
   {\mathchoice
       {{#1}\brackets*{#2}}
       {{#1}[ #2 ]}
       {{#1}[ #2 ]}
       {{#1}[ #2 ]}
   }
}
\def\renderwithdist#1#2#3{%
   \@ifnextchar\bgroup
   {\superfancyrender{#1}{#2}{#3}}
   {\mathchoice
      {\underset{#2}{#1}\brackets*{#3}}
      {{#1}_{#2}[ #3 ]}
      {{#1}_{#2}[ #3 ]}
      {{#1}_{#2}[ #3 ]}
     }
}
\def\superfancyrender#1#2#3#4#5{
   \ensuremath{\mathchoice
      {\underset{#1}{{#1}}\left#4 #3 \right#5}
      {{#1}_{#2}#4 #3 #5}
      {{#1}_{#2}#4 #3 #5}
      {{#1}_{#2}#4 #3 #5}
   }
}
\makeatother

 \newcommand\SetSymbol[1][]{%
     \nonscript\:#1\vert
     \allowbreak
     \nonscript\:
     \mathopen{}}
  \DeclarePairedDelimiterX\Set[1]\{\}{%
     \renewcommand\given{\SetSymbol[\delimsize]}
     #1
}

\newcommand{\tr}{\mathrm{tr}}

\newcommand{\fkg}{\mathfrak{g} }
\newcommand{\fkh}{\mathfrak{h} }

\definecolor{asparagus}{rgb}{0.53, 0.66, 0.42}
\definecolor{cambridgeblue}{rgb}{0.64, 0.76, 0.68}
\definecolor{celadon}{rgb}{0.67, 0.88, 0.69}
\definecolor{charcoal}{rgb}{0.21, 0.27, 0.31}

\definecolor{cadmiumgreen}{rgb}{0.0, 0.42, 0.24}

\usepackage{graphicx}
\usepackage{booktabs}

\usepackage[top=1in, bottom=1in, left=1.25in, right=1.25in]{geometry}
\usepackage{tcolorbox}
\usepackage{lipsum}

\newcommand{\snote}[1]{\textcolor{asparagus}{ {\textbf{(Sourya: #1)}}}}
\allowdisplaybreaks
\DeclareUnicodeCharacter{00A0}{ }
\begin{document}
	
	\title{A General Framework for Low Soundness Homomorphism Testing}
	\author{Tushant Mittal\thanks{Stanford University, \; {\tt{tushant@stanford.edu}}. TM is a postdoctoral fellow supported by the NSF grants CCF-2143246 and CCF-2133154.} \and
        Sourya Roy\thanks{The University of Iowa, \; {\tt{sourya-roy@uiowa.edu}}. SR was partially supported by the Old Gold Summer Fellowship from The University of Iowa.}     }
	
%	\author{
	%	Tushant Mittal\thanks{{\tt University of Chicago}. {\tt tushant@uchicago.edu}. Supported by NSF grant CCF-2326685.} \and
	%	Sourya Roy\thanks{{\tt The University of Iowa}. {\tt sourya-roy@uiowa.edu}.}    \and
	%}
	
	\date{\today}

	\date{}
	
	\maketitle
	%\vspace{-1.35cm}
	\draftbox

%, we use $H\backslash S_n \slash H$
We introduce a general framework to design and analyze algorithms for the problem of testing homomorphisms between finite groups in the low-soundness regime. 

\medskip

In this regime, we give the first constant-query tests for various families of groups. These include tests for:
(i) homomorphisms between arbitrary cyclic groups, (ii) homomorphisms between any finite group and $\Z_p$, (iii) automorphisms of dihedral and symmetric groups, (iv) inner automorphisms of non-abelian {finite simple groups} and {extraspecial groups}, and (v) testing linear characters of $\GL_n(\F_q)$, and finite-dimensional Lie algebras over $\F_q$. We also recover the result of Kiwi [TCS'03] for testing homomorphisms between $\F_q^n$ and $\F_q$. 

\medskip

Prior to this work, such tests were only known for abelian groups with a constant maximal order (such as $\F_q^n$). No tests were known for non-abelian groups.

\medskip

As an additional corollary, our framework gives combinatorial list decoding bounds for cyclic groups with list size dependence of $O(\ep^{-2})$ (for agreement parameter $\ep$). This improves upon the currently best-known bound of $O(\ep^{-105})$ due to Dinur, Grigorescu, Kopparty, and Sudan~[STOC'08], and Guo and Sudan~[RANDOM'14]. 
% \begin{tcolorbox}[title=To Do, colframe=charcoal, colback=white]
% 	\begin{itemize}
%% 		\item Compute $\frac{	\sum_{\phi\in \Hom(G,H)}  \tilde \alpha_{\phi}^k}{	\sum_{\phi\in \Hom(G,H)}  \tilde\alpha_{\phi}^{k-1} } 
%% 		$ for $G=\Z_2^n$ to $H=\Z_2$ for $k=\text{odd}$ and $k\geq 3$.  		
%% 		\item $\hom \text{ testing for } G=\F^n_p$ and $H=\F_p$
%% 		\item $\hom \text{ testing for } G=\F^n_p$ and $H=\F^m_p$
%% 		\item $\hom \text{ testing for cyclic } G$ 
%% 		\item$\hom \text{ testing for } G=G_1\times G_2$ for cyclic $G_1,~G_2$ 
% 		\item $\Aut(G)$ testing for $D_{n}$
% 		\item Derandomization of tests for $G=\F^n_p$ and $H=\F_p$
% 	\end{itemize}
% \end{tcolorbox}
\thispagestyle{empty}
\addtocounter{page}{-1}
 \clearpage
 \tableofcontents
% \listoftodos[To-Do]
 \thispagestyle{empty}
 \addtocounter{page}{-1}
\newpage
 \setcounter{page}{1}
\section{Introduction}
The problem of \textit{homomorphism testing} has been extensively studied in the theoretical computer science literature~\cite{BCH+95, BLR90, Kiwi03,Sam07,HW03}.  A primary motivation for studying this question is its relevance to the theory of \textit{probabilistically checkable proofs} \cite{BSVW03,BK21} and \textit{locally testable codes}. Additionally, there has been an interest in studying such tests in quantum complexity, for example, \textit{entanglement testing}~\cite{NV17} involves homomorphism testing which played an important role in the proof of $\mathrm{MIP}^*\text{=}\mathrm{RE}$~\cite{MIP21}. 
%Additionally, such non-Abelian tests have been used for constructions of better PCPs~, and hardness of approximation results~\cite{BKM22}.   

%In the language of coding theory, the problem of homomorphism testing is $f$ is close to any codeword in the homomorphism code consisting of the set of homomorphisms from $G$ to $H$, $\Hom(G,H)$. 

In this work, we study this homomorphism testing problem in the context of general finite groups. To make the discussion precise, we begin by formally describing the setup. Let $G$ and $H$ be two finite groups. Denote by $\Hom(G,H)$, the set of all homomorphisms from $G$ to $H$, \ie a function such that, $f(xy) = f(x)\cdot f(y)$ for each $x,y \in G$.

\begin{definition}
	$\Hom(G,H)$ is $(k, \delta, \ep)$-testable if there exists an algorithm (test) that, given oracle access to a function $f: G\to H$, makes $k$ queries to it, and satisfies the following: 
	\begin{itemize}
		\item{(Completeness)} If $f$ is a homomorphism, the test passes with probability $1$.
		\item{(Soundness)} If test passes with probability $\delta$ (over the choice of queries), then there exists a homomorphism $\phi$ such that 
		 $\agr(f,\phi)~:=~\Pr{x\sim G}{f(x)=\phi(x)} ~\geq~ \ep(\delta)$. 
	\end{itemize}
\end{definition}
 As an example, the famous Blum--Luby--Rubinfield~\cite{BLR90} (BLR) test samples a random pair $x,y\sim G$ and checks if $f(x)\cdot f(y) = f(xy)$. This test shows that $\Hom(\F_2^n, \F_2)$, also known as the \textit{Hadamard code}, is $\parens{3,\delta, \delta}$-testable. We are interested in identifying finite groups, $(G,H)$, for which the set $\Hom(G,H)$ is testable and designing such tests.

\paragraph{High Soundness Regime} It is much easier to construct a test that only guarantees soundness when a function passes the test with a probability much larger than the test passing probability of a random function. This regime is known as the high soundness or the unique decoding regime, as there is often a unique homomorphism that agrees with the input function. There are many results in this regime and in particular, Ben Or-Coppersmith-Luby-Rubinfeld~\cite{BCLR07}, showed that $\Hom(G,H)$ is $(3,\delta, 1-\frac{\delta}{2})$-testable for $\delta > \frac{7}{9}$ for any finite groups $G,H$. Moreover, the test is the same as the BLR test.

\paragraph{Low Soundness Regime} It is significantly more difficult to design and analyze tests in the low soundness (\textit{list decoding}) setting when the test passing probability can be arbitrarily small, and the function has a tiny agreement with many homomorphisms. As a sharp contrast to the result of~\cite{BCLR07}, the only known cases for which the BLR test has been analyzed in this low soundness setting are: 
 (i) $(\F_p^n, \F_p)$ for some prime $p$, by H{\aa}stad and Wigderson~\cite{HW03}, and (ii) $(\F_p^n, \F_p^m)$ by Samrodnitsky\footnote{For $p=2$, this setting is equivalent to the Freiman-Rusza conjecture for which improved bounds were proven in the breakthrough works of~\cite{Sanders10, GGMT23}.}~\cite{Sam07} which can be generalized to the setting of $G= \Z_p^{n_1}\oplus \dots \oplus \Z_{p^r}^{n_r}$, where $r = O(1)$.

%This is because the Fourier-theoretic proof for the case of $p=2$ 

\paragraph{Issues with BLR} The high-soundness result of~\cite{BCLR07} cannot be hoped to generalized to the low soundness setting, as they also give the following counterexample:  for any $r \geq 3$, there exists a function $f: \Z_{3^r} \to \Z_{3^{r-1}}$ that passes the BLR test with probability $\frac{7}{9}$ but agrees with any homomorphism at most $3^{-(r-1)}$-fraction of points. This demonstrates that BLR fails catastrophically, even for cyclic groups, in the low-soundness regime. 

Moreover, even in cases for which BLR works, it is not always known to yield the best agreement guarantee. For instance, \cite{HW03} showed that the BLR test for $(\F_p^n, \F_p)$ achieves an agreement guarantee of $\ep(\delta) = \frac{1+\delta}{p}$. Hence, even when the function passes with probability $1$, the guarantee is small for large $p$. This was remedied by Kiwi~\cite{Kiwi03} by giving a different test\footnote{Grigorescu,
Kopparty and Sudan~\cite{GKS06}, and Gopalan~\cite{Gopalan13} gave an alternate Fourier analytic proof of Kiwi's result.}  which showed that $\Hom(\F_q^n, \F_q)$ is $(3,\delta, \delta)$-testable.

%Therefore, one has to devise new tests to tackle a more general family , and obtain optimal bounds  
% His proof used combinatorial properties of the Krawtchouk polynomials and the weight distribution of the dual codes. 
% However, in general, when the target group $H$ does not embed easily into $\C$, one cannot directly rely on Fourier analysis. For instance, the testing result for $(\F_2^n, \F_2^m)$ uses machinery from additive combinatorics such as the Balog-Szemer\'edi-Gowers lemma that has no Fourier proof. 
 
The above discussion shows that new tests and techniques must be devised to handle new families of groups and/or obtain optimal parameters. Additionally, it is desirable to have tests that can provide an improved soundness guarantee by using more queries. 

%In this work, we make progress towards this goal by giving a general unified framework to define and analyze $k$-query homomorphism tests for general groups.

% $\ep(\delta)$ improves with larger number of queries. 
%This is the starting point of our work. 

\subsection{Our Contribution}

We take a first step towards this by defining a general testing framework and using it to get testability results for a variety of groups. In particular, we give the first tests for classes of non-abelian groups in the low-soundness setting. Our meta-test is the following. Let $G,H$ be finite groups, $k$ be an integer, and let $\cD_k$ be a distribution on $G^k$.  
 
 \vspace{1em}
\begin{tcolorbox}[colframe=teal, colback=white, title={$\test_k(G,H, \cD_k)$}]
	\begin{itemize}
		\item Sample $(x_1,\dots, x_k) \sim \cD_k \subseteq G^k$.
		\item Return $1$ if and only if there exists\footnotemark a homomorphism (or automorphism) $\phi$ such that $f(x_i) = \phi(x_i)$ for $i \in [k]$. 
%		\item return $1$; otherwise: return $0$
	\end{itemize}
%$D_{\ker}(\vec x)=\frac{|\ker(\Gamma_{\vec x})|}{         \sum_{{\vec x \in G^k~:~ {\sf H}_{\vec x}\neq H^k}}   |\ker(\Gamma_{\vec x})| }$
\end{tcolorbox} \vspace{0.5em}

\footnotetext{For almost all the groups we study, there is a simple and efficient way to check this. }

%A key conceptual difference from previous works is that we do not start from a prescribed test and then analyze it, but rather start from an (approximating) expression for $\max_{\phi \in \Hom(G,H)} \agr(f,\phi)$, \ie the maximum agreement of $f$ with a homomorphism from which the required test naturally emerges.

We explain the framework in~\cref{sec:tech_over}, but briefly summarize its salient features:
\begin{itemize}
	\item \textbf{Defining the Distribution $\cD_k$} -- A key feature of the framework is that this distribution naturally emerges from an (approximating) expression for the soundness, \ie $\max_{\phi \in \Hom(G,H)} \agr(f,\phi)$.    The BLR test uses the uniform distribution over $\{(x_1,x_2,x_3) \mid x_1x_2x_3 = 1\} $ as $\cD_k$, regardless of the groups $G,H$. In contrast, our distribution takes into account the group-theoretic data. In the special case of $(\F_2^n, \F_2)$, our distribution coincides with the one in BLR, but for other groups, they are quite different. For example, in the case of cyclic groups, our distribution (roughly) weighs elements based on their order and is not uniform. This difference intuitively explains why BLR fails for cyclic groups, but our test works. 		    
		\item \textbf{Distance Approximation and List Decoding} -- Our analysis works by giving an exact expression for the $k^{\mathrm{th}}$-moment,  $\sum_\phi \agrHom^k$. This can be seen as a \enquote{degree-$k$ variant} of the general {Johnson bound} (which is for $k=2$). Such an expression not only yields a soundness guarantee but also implies a bound on  (i) the number of \enquote{large-agreement homomorphisms}, \ie combinatorial list size bounds for homomorphism codes, and (ii) the largest possible agreement that gives a tight control on the distance of the input function $f$ to the property of being a homomorphism.  This task of distance approximation was introduced in~\cite{PPR06}, where they show that approximating distance implies \textit{tolerant tests}. 
		\item  \textbf{Avoiding Fourier Analysis} -- Our technique works entirely in the \enquote{physical space} by reducing the soundness analysis to the computation of certain group-theoretic constants. For some classes of groups (such as finite simple groups), these constants have been studied in the literature. This is helpful when the target group $H$ does not embed easily into $\C$, and one cannot directly rely on Fourier analysis. %	\item Thus, we get an alternate proof of Kiwi's result that is different from the original one that uses weight distribution bounds, and subsequent ones that use Fourier analysis  (\cite{GKS06, Gopalan13}).  
	\item 	\textbf{Query vs Soundness Tradeoff} -- The above test works for any (large enough) number of queries $k$, and the analysis shows that the test guarantee improves exponentially with the number of queries, $\ep(\delta) \approx \delta^{\frac{1}{k}}$, giving us a smooth tradeoff between query complexity and the soundness guarantee.  
\end{itemize}

%We now mention our results concretely

\subsubsection{Homomorphism Testing}

 We now summarize our results for homomorphism testing over various classes of finite abelian groups.  
 To the best of our knowledge, all the tests here are novel except for the $3$-query test for $\Hom(\F_q^n, \F_q)$ due to~\cite{Kiwi03}.

%\begin{table}[h]
%	\begin{center}
%		\small
%		\begin{tabular}{cccc}
%			\toprule
%			Work & Domain $G$ & Range $H$   
%			\\[2pt]
%			\midrule
%			\multicolumn{3}{c}{High Soundness}\\[1.9pt]
%			\midrule
%			\cite{BLR90} & $\Z_2^n$ &$\Z_2$ 	\\
%			\cite{BCLR07} & $G,H$ any finite groups  & $2\log |G|$	\\
%			\cite{SW04} & $G,H$ any finite groups & $(1+o(1))\log |G|$	\\[2pt]  \midrule
%			\multicolumn{3}{c}{Low Soundness}\\[1.8pt]
%			\midrule
%			\cite{BCH+95} & $G = \Z_2^n,\, H = \Z_2$ & $2\log |G|$	\\
%			\cite{Kiwi03} & $G = \Z_p^n,\, H = \Z_p$  & $(2+o(1))\log |G|$	\\
%			\cite{BSVW03} & $G = \Z_p^n,\, H = \Z_p$  & $(1+o(1))\log |G|$	\\
%			\cite{Sam07, Sanders10, GGMT23} & $\Z_2^n \to \Z_2^m$ & \\
%			\hline
%		\end{tabular}\caption{A summary of prior works on homomorphism testing}
%	\end{center}
%\end{table}

%\begin{table}[h]
%	\begin{center}
%		\small
%		\begin{tabular}{cccc}
%			\toprule
%			Work & $G$ & $H$ 
%			\\[2pt]
%			\midrule
%			\cite{BCH+95, Kiwi03,BSVW03} & $\Z_q^n$ & $\Z_q$	\\
%			\cite{Sam07, Sanders10, GGMT23} & $\Z_2^n$& $\Z_2^m$  \\
%						\midrule
%		\multicolumn{3}{c}{Our Results}\\[1.9pt]
%\midrule
%			 & $\Z_q^n$ & $\Z_q$	\\
%			 & $\Z_n$ & $\Z_m$	\\
%			 & Cyclic  & $\Z_q$	\\
%			 & $\Z_q^n$ & $\Z_q$	\\
%			\cite{Sam07, Sanders10, GGMT23} & $\Z_2^n$& $\Z_2^m$  \\
%%			\multicolumn{3}{c}{New Results}\\[1.9pt]
%			\hline
%		\end{tabular}\caption{A summary of prior works on homomorphism testing}
%	\end{center}
%\end{table}

%Note that for a prime $p$, every abelian group is a $^{\geq}p$-group.

\begin{theorem}[Summary of Homomorphism Testing]\label{theo:main}
	For each pair of groups $G,H$ as listed in~\cref{tab:main}, and correspondingly allowed integers $k$, $\Hom(G,H)$ is $(k,\delta, \ep(\delta))$-testable. The test is $\test_k(G,H,\cD_k)$, for an explicitly defined distribution $\cD_k$ on $G^k$. Additionally, we also get an upper bound of $\max_\phi \agrHom \leq O(\delta^{\frac{1}{k}})$, for any appropriate $k$ as in the table.  
	\end{theorem}

\begin{table}[h]
	\begin{center}
		\small
		\begin{tabular}{ccccc}
			\toprule
		Result &	$G$ & $H$ & Query & Soundness  $\brackets{\eps(\delta)}$ 
			\\[2pt]
			\midrule
		{\cref{theo:vspace}} 	&  $\F_q^n$& $\F_q$ &
		 {Odd $k \geq 3$} & \multirow{2}{*}{${\frac{1}{q} - O(\frac{1}{q^n}) + \parens{\frac{q-1}{q}}\cdot  \parens[\big]{\frac{q\delta-1}{q-1}}^{\frac{1}{k-2}} }$}	\\[4pt]
		 	\cite{Kiwi03} 	&  $\F_q^n$& $\F_q$ &  $k = 3$ &
%		 \multirow{1}{*}{ $k = 3$} & {${\frac{1}{q} - O(\frac{1}{q^n}) + \parens{\frac{q-1}{q}}\cdot  \parens[\big]{\frac{q\delta-1}{q-1}} }$}	
\\[3pt]
		\midrule
	{\cref{thm:cyclic_gen}}		&  $\Z_n$ & $\Z_m$	& {Any $k \geq 4$} & {$ \parens[\big]{\zeta(2)^2\cdot\delta}^{\frac{1}{k-3}}$}\\
			\midrule
			\cref{theo:fq_vspace} & $\F_q$  & $\F_q^n$	&  {Any $k \geq 2$} &  {$\frac{q-1}{q}\cdot \delta^{\frac{1}{k-1}}$} \\
%			 & $\Z_p$  & Any Abelian group	&\\
			 \midrule
	 {\cref{thm:cyclic_bounded}} & $\Z_{p^r}$ & \makecell{Abelian group of\\$p$-rank $\leq t$}	& {Any $k \geq t+2$} & {$\parens[\Big]{\frac{(p-1)^2}{p^2}\cdot \delta}^{\frac{1}{k-t-1}}$} \\[2pt]
			\hline\\
		\end{tabular}
			\vspace{-0.3em}
\begin{minipage}{\linewidth}
\centering \small
Note: The $p$-rank of an Abelian group is the number of cyclic groups of order a power of $p$ in its decomposition, see~\cref{fact:decomp}.
% of  $\GL_n(q)$ is the group of $n\times n$ invertible matrices over $\F_q$.
% and $D_{2p}$ is the Dihedral group of order $2p$ for some prime $p$.
\end{minipage}
\caption{A summary of our results on homomorphism testing. }\label{tab:main}
	\end{center}
\end{table}
\vspace{-2em}
\paragraph{Combinatorial List Decoding}
While the focus of our work is not list decoding, our technique yields stronger bounds than currently known for some groups. This is because our analysis gives bounds on $\sum_\phi\agr(f,\phi)^k$, and the list size then immediately follows. 

The work of Dinur, Grigorescu, Kopparty, and Sudan~\cite{DGKS08}, and later Guo and Sudan~\cite{GS14} gave a list size bound of $O\parens[\big]{\eps^{-105}}$ that works for every pair of abelian groups (and in fact \enquote{supersolvable} groups). However, their bound does not improve even when $H$ is cyclic. We prove a much better bound of $\eps^{-2}$ for cyclic groups of prime power and $\eps^{-3}$ for general cyclic groups.    

\begin{theorem}[List Decoding for Cyclic groups]
	Let  $G=\Z_{p^r}$  and $H = \Z_{p^s}$ be cyclic groups. Let $f: G\to H$ be any function.  Then the following holds:
		\[ \abs[\big]{\{ \phi\in \Hom(G,H)  ~:~\agr(f,\phi) ~\geq~ \ep  \}} ~\leq~ \parens[\Big]{\frac{p}{p-1}} \cdot  \frac{1}{\ep^{2}}.\]
In general, we get a list size bound of ${2\ep^{-(t+1)}}$ when $H$ is an abelian group of $p$-rank $t\geq 1$. Additionally, for any integers $n, m \geq 1$, we get the following list size bound:  			
	\[ \abs[\big]{\{ \phi\in \Hom(\Z_n,\Z_m)  ~:~\agr(f,\phi) ~\geq~ \ep  \}} ~\leq~  \frac{\zeta(2)^2}{\ep^{3}}\;,\]
	where $\zeta(2) = \frac{\pi^2}{6}$, is the Riemmann-Zeta function.
\end{theorem}

\subsubsection{Automorphism Testing over Non-Abelian Groups}
A very important class of homomorphisms arises from automorphisms of groups. Our next set of results concern testing automorphisms and inner automorphisms over various families for finite non-abelian groups. We quickly recall the relevant definitions,
\begin{align*}
	\Aut(G) ~&=~  \braces{\phi \in \Hom(G,G) \mid\phi\; \text{is bijective} } , \\ 
	\Inn(G) ~&=~  \braces{\phi_g,\; g \in G  \mid \phi_g(x) = gx\inv{g} } ~\subseteq~ \Aut(G) . 
\end{align*}

The set of inner automorphisms is easier to work with as the maps are very explicit. Moreover, for many groups, the inner automorphisms capture most of the automorphisms. For example, for the symmetric group $\Sym_n$, all the automorphisms are inner (\cite{Seg40}) (for $n \neq 6$). The following theorem consolidates all our automorphism testing results. 

\begin{theorem}[Automorphism Testing (Summary of \cref{thm:dihedral_aut,theo:main_aut2})]\label{theo:main_aut}
	The following results hold by using (a modification of) $\test_k(G,G,\cD_k)$ for an explicitly defined distribution $\cD_k$ on $G^k$ \textup{:} 
	\begin{itemize}
		\item For the family of dihedral groups, $D_{2p}$ ($p$ prime), $\Aut(D_{2p})$ is $\parens[\Big]{k,\,\delta,\, \frac{1}{2} \delta^{\frac{1}{k-2}} }$-testable for every $k \geq 3$.
		\item For the family of symmetric groups, $\Aut(\Sym_n) = \Inn(\Sym_n)$ is $\parens[\Big]{k,\,\delta,\,  \delta^{\frac{1}{k-2}} - o_n(1)}$-testable for every $k \geq 3$, and $n \neq 6$. 
		\item For every non-abelian finite simple group $G$, $\Inn(G)$ is $\parens[\Big]{k,\,\delta,\,  \delta^{\frac{1}{k-2}} - o_{|G|}(1)}$-testable for every $k \geq 4$.
		\item For any extraspecial group $G$ of order $p^r$, $\Inn(G)$ is $\parens[\Big]{k,\,\delta,\,  \delta^{\frac{1}{k-r}} - o_{p}(1)}$-testable for every $k \geq r+1$. In particular, for the family of Heisenberg groups ($H_p$) (the group of $3\times 3$ unitrianguar matrices over $\F_p$),   $\Inn(H_p)$ is   $\parens[\Big]{k,\,\delta,\,  \delta^{\frac{1}{k-3}} - o_{p}(1)}$-testable for any $k \geq 4$.
	\end{itemize}
	Additionally, we also get an upper bound of $\max_\phi \agrHom \leq O(\delta^{\frac{1}{k}})$, for any group $G$ and $k$ as above, except the dihedral group.
	\end{theorem}

\subsubsection{Lifting Homomorphism Tests} 
We give a general method to extend results for testability of known $\Hom(G,H)$ to those of $\Hom(\tG, H)$ in cases when $\Hom(\tG, H)$ factors through $\Hom(G,H)$. 
\[\vcenter{
\hbox{
\begin{tikzcd}
\tG \arrow[d, two heads, "\pi"] \arrow[rd, dashed, " \forall\tilde{\varphi}"] &  \\
 G \arrow[r,"\exists! \varphi" below] & H
 \end{tikzcd}
}}\]	
For instance, when $H$ is abelian, every homomorphism from $G$ to $H$ factors through an \enquote{abelianization} of $G$, which is defined as $G/[G,G]$ where $[G,G]$ is the subgroup generated by $\angles{xyx^{-1}y^{-1} \mid x,y \in G}$. This also works when we have a more general structure like a \textit{Lie algebra}. We quickly define the notion of a character  over a Lie algebra.

\paragraph{Lie algebra characters}

 A finite-dimensional Lie algebra, $\mathfrak{g}$, is a finite-dimensional vector space over a field $\F$ with a Lie bracket $[\cdot, \cdot ] : \fkg\times \fkg\to \fkg$ which is a bilinear map such that 
\[
 [x,x] = 0 \quad \text{and} \quad  [x,[y,z]] + [y,[z,x]] + [z,[x,y]] = 0, \;\; \forall \,  x,y,z \in \fkg.
\]

% \begin{enumerate}
% 	\item $[x,x] = 0$, and
% 	\item  $[x,[y,z]] + [y,[z,x]] + [z,[x,y]] = 0 $, for any $x,y,z \in \fkg$
% \end{enumerate}
A \textit{linear character} of a Lie algebra is a linear map $\phi: \fkg\to \F$ such that  $f([x,y]) = 0$ for every $x,y \in \fkg$. Therefore, it is a linear map between vector spaces subject to an additional constraint. For example, let $\mathfrak{gl}_n(q) = \F_q^{n\times n}$, the vector space of $n\times n$-matrices. The bracket is defined as $[x,y] := xy-yx$, where the multiplication is matrix multiplication. A character of $\mathfrak{gl}_n$ is a linear map with the property that $f(xy) =f(yx)$ for every pair of matrices $x,y$. 

Character testing has also been studied in the literature~\cite{BFL03,MR15,OY16,GH17,MR24} for general groups (and more general representations).  However, to the best of our knowledge, all known results work in the $L^2$-metric, \ie the soundness guarantee is in terms of $\norm{f(x)-\phi}_2^2$. Our result gives the first character testing results for non-abelian groups in the Hamming metric.

\begin{theorem}[Lifting results]\label{theo:cyclic_1}
	For the groups/Lie algebras mentioned in~\cref{tab:main2}, one can obtain testing results by lifting from their base code. Moreover, the queries and soundness guarantees are identical to those of the base code.
\end{theorem}

\begin{table}[ht]
	\begin{center}
		\small
		\begin{tabular}{cccc}
			\toprule
		Result &	$G$ & $H$ & Base Code 
			\\[2pt]
			\midrule
\multirow{2}{*}{{\cref{cor:lift_lie}}}		&  Any finite group & \multirow{2}{*}{$\F_p$} &
		\multirow{2}{*}{ $\Hom(\F_p^{n}, \F_p)$  } 	\\[2pt]
%		Prime power $q$		&  Any $\qgroup$-group & $\F_q$ &  & 	\\[2pt]
				&  Any Lie Algebra over $\F_p$  &  &   	\\[2pt]
		\midrule
	{\cref{theo:GL_n}}		&  $\GL_n(q), \;q \neq 2$  & $\F^*_q$	& $\Hom(\Z_{q-1}, \Z_{q-1})$ \\
		%			 & $\Z_p$  & Any Abelian group	&\\
			 \midrule
	 	{\cref{theo:gln}}		&  $\mathfrak{gl}_n(q)$  & $\F_q$ & $\Hom(\Z_{q}, \Z_{q})$	\\[2pt]
			\hline\\
		\end{tabular}
			\vspace{-0.3em}
\begin{minipage}{0.8\linewidth}
\centering \small
Note: For the first result, $n$ is the $p$-rank of $G$, \ie the $p$-rank of its abelianization, or the rank of the Lie algebra (see~\cref{fact:lie_ab}).
% \\Refer to~\cref{fact:lie_ab} for a definition of rank of a Lie algebra. 
% and $D_{2p}$ is the Dihedral group of order $2p$ for some prime $p$.
\end{minipage}
\caption{A summary of our lifting results.}\label{tab:main2}
	\end{center}
\end{table}
\vspace{-1em}

\begin{remark}
In coding theory terms, the lifted homomorphism code (up to permutation) is the base code tensored with the repetition code of length $m = (\abs{\ker(\pi)})$. Thus, one could alternatively design a test from this perspective, or appeal to results on local testability of tensor codes such as~\cite{DSW06}. However, our approach yields the result easily by allowing us to mechanically reuse the analysis of the base code. 	
\end{remark}
%\newpage
%\medskip
%\medskip

\subsection{Technical Overview}\label{sec:tech_over}

For a function $f$ and a homomorphism $\phi$, let $\agr(f, \phi)$ be the agreement between these functions, \ie the fraction of inputs on which they agree. We wish to estimate $ \max_{\phi} \agr(f, \phi)$. To do this, we will define a distribution $\cal{F}$ on $\Hom(G,H)$. Clearly,
\begin{equation}
\label{eq:approx} \max_{\phi\in \Hom(G,H)} \agr(f, \phi) ~\geq~  \Ex{\phi\sim \cal{F}}{\agr(f, \phi)}	\;.
\end{equation}

For this approximation to be useful, the distribution must have a large mass on homomorphisms that agree significantly with $f$. A natural choice for such a distribution is  $\Pr{\phi} \propto \agr(f, \phi)^k$ for some positive integer $k$. Note that this is general and agnostic to the choice of groups (or even the fact that the code is a homomorphism code).

%For the above to work best, the ideal choice for $\cal F$ will be a distribution that is only supported on the homomorphisms which have maximum agreement with $f$. However, 

This expectation then becomes:
\[ \Ex{\phi\sim \cal{F}}{\agr(f, \phi)} ~=~ \frac{\sum_\phi \agr(f, \phi)^{k+1}}{\sum_\phi \agr(f, \phi)^{k}} \;.
\] 

The next step is to estimate this expectation via a test that queries $f$ at only a few points. We do this by reinterpreting the expression for the $k^{\mathrm{th}}$ powers algebraically, using the knowledge that the code is a homomorphism code. The key point of the framework is that after this reinterpretation, the definition of a test pops quite intuitively, such that 
\begin{equation}
\label{eq:test_k} {\sum_{\phi\in \Hom(G,H)} \agr(f, \phi)^{k}} ~\propto~  {\Pr{f \;\text{ passes }\; \mathsf{Test}_k }}. 
\end{equation} 

Once we have this, the main testing result is a simple calculation. We now explain our algebraic reinterpretation of this expression and the test that emerges from it. The key to our method is the following evaluation map\footnote{We thank MO user \textit{t3suji} whose comment on~\cite{suji} was the inspiration for us to study this map. }.

\paragraph{The evaluation map} One important way in which this framework utilizes the structure of $\Hom(G,H)$, is that for any fixed tuple $\vec{x} = (x_1, \dots, x_k) \in G^k$, there is a naturally associated evaluation map,
% $\Gamma_{\vec{x}} : \Hom(G,H) \to H^k$ which just evaluates the homomorphism 
\[
\Gamma_{\vec{x}} : \Hom(G,H) \to H^k, \; \Gamma_{\vec{x}} (\phi) = (\phi(x_1), \dots, \phi(x_k)). \] 

While this map could be defined for any subset of functions (and not necessarily homomorphisms), for the set of homomorphisms, the map is $N$-to-one on its image, where $N = \abs{\ker(\Gamma_{\vec{x}})}$. This crucial property immediately implies that,
\[
\sum_\phi \agr(f, \phi)^{k} ~=~   \Ex{\vec{x} \sim G^k}{ \indicator{f(\vec x) \in \mathrm{Im}(\Gamma_{\vec x})} \cdot \abs{\ker(\Gamma_{\vec x})} } . \] 

The right-hand side can now be interpreted as a test wherein a tuple is sampled with a weight $\propto |\ker \Gamma_{\vec x} |$, and the indicator $\indicator{f(\vec x) \in \mathrm{Im}(\Gamma_{\vec x})}$ tests if there exists a homomorphism $\phi$ with which $f$ agrees on the entire tuple $\vec{x}$. This test trivially passes if $\mathrm{Im}(\Gamma_{\vec x})= H^k$  and thus we exclude such tuples. The final test is then, 
 \vspace{1em}
\begin{tcolorbox}[colframe=teal, colback=white, title={$\test_k(G,H)$}, label={tvan}]
	\begin{itemize}
		\item Sample $(x_1,\dots, x_k) \propto |\ker \Gamma_{\vec x} |$ subject to $\mathrm{Im}(\Gamma_{\vec x}) \neq H^k$.
		\item If $(f(x_1),\dots, f(x_k))\in \mathrm{Im}(\Gamma_{\vec x}) $: return $1$; otherwise: return $0$.
	\end{itemize}
%$D_{\ker}(\vec x)=\frac{|\ker(\Gamma_{\vec x})|}{         \sum_{{\vec x \in G^k~:~ {\sf H}_{\vec x}\neq H^k}}   |\ker(\Gamma_{\vec x})| }$
\end{tcolorbox}

\paragraph{Analyzing and adjusting the test}

The above recipe works as it is for a given pair of groups if the following hold:
\begin{enumerate}
	\item  Our approximation, \ie ~\cref{eq:approx}, is not too weak, and
	\item  The test we have defined indeed approximates $\sum_\phi \agrHom^k$ well.
\end{enumerate}

Among the cases we analyze, this happens for cyclic (and cyclic-like) groups, and our results in~\cref{sec:cyclic} directly use this test. However, for $\Hom(\F_q^n, \F_q)$, the first condition does not hold. This is remedied by using a shifted variant of $\agr(f,\phi)$. Moreover, verifying the second point, \ie checking if \cref{eq:test_k} holds can be difficult in general, and another trick we employ is to define the test on a subset of all $k$-tuples that makes the analysis more manageable. We wish to emphasize that~\hyperref[tvan]{$\test_k(G,H)$} gives a starting point from which to derive a test for a general pair of groups. 

%In this work, we carry out this investigation for cyclic groups, ...

%\paragraph{Capturing the usual BLR} We now demonstrate that our test is a slightly more stringent version of BLR test. Let $G = \Z_2^n$ and $H = \Z_2$. Then $\agr(f,\phi) = \frac{1+\widehat{f}(\phi)}{2}$. 
%
%This suggests that we should instead work with a shifted variant, 
%\[\widetilde{\agr}(f,\phi)~=~\frac{|H|\agr(f,\phi)-1}{|H|-1}\]
%This scaling can be derived by setting that $\widetilde{\agr}(\phi,\phi) = 1$ and $\widetilde{\agr}(f,\phi) = 0$ for a random function $f$. Using this scaled version for $G=\Z_2^n \text{ and } H=\Z_2$, we have $\tilde{\alpha}_{\phi}=\widehat{f}(\phi)$. 
%
%\tnote{Write stuff about }
%
%\paragraph{Adjusting the agreement} 
%
%\subsection{Our Results}
%
%\subsection{Related Work}
%dd

\subsection{Outline} 
We start in \cref{sec:prelim} by summarizing basic definitions and 
some of the notation used throughout the paper. The general testing framework developed in 
~\cref{sec:test} captures the core methodology that is used throughout the paper to prove our main results. 

As our first application, we use the proposed framework in~\cref{sec:cyclic} to establish the testing and list decoding results for cyclic groups. In particular,  in~\cref{subsec:cyclic_bounded}, we prove the testing result, ~\cref{thm:cyclic_bounded}, for $G=\Z_{p^r}$ and $H=\text{abelian group of bounded } p$-rank; here, we also prove the list decoding theorem, ~\cref{theo:list_size}. In \cref{subsec:cyclic_gen}, we generalize to the case when $G$ and $H$ are arbitrary cyclic groups. 

We focus on vector spaces in~\cref{sec:vspace}. The main result of this section is~\cref{theo:vspace} that allows us to test functions from $\F_q^n \to \F_{q}$. In the same section, we also derive a testing result (\cref{theo:fq_vspace}) for $\Hom(\F_{q},\F^n_{q})$. 

In~\cref{sec:auto} and ~\cref{sec:beyond}, we consider the non-abelian setting. \cref{sec:auto} focuses on testing for closeness to an automorphism or an inner automorphism. We look at dihedral groups (\cref{thm:dihedral_aut}), symmetric group, quasirandom groups, and more generally, finite simple groups (\cref{theo:main_aut2}). Finally, in~\cref{sec:beyond}, we prove a general lifting theorem. This allows us to prove our character testing results (\cref{theo:GL_n,theo:gln}), and other lifted results~\cref{cor:lift_lie}. 
\subsection{Prelims}\label{sec:prelim}
%\begin{fact}[Hall's theorem]
% 	\label{fact:hall}
% 	content...\todo{add Hall's theorem here}
%\end{fact}

 \begin{fact}[$p$-components of Abelian groups]
		\label{fact:decomp}
Every finite abelian group decomposes into $G = \oplus_p G_p$ where $G_p = \oplus_i \Z_{p^{b_i}}$ is the $p$-component of $G$. The $p$-rank of $G$ is the number of summands in $G_p$. Moreover, 
$\Hom(G,H) \cong \oplus_p\Hom(G_p,H_p)$.
\end{fact}
 
 \begin{fact}[Cyclic Homomorphisms]
		\label{fact:cyc_homs}
	Let $G=\Z_{p^r}$ and $H$ be any abelian group with a $p$-component $\oplus_i \Z_{p^{b_i}}$. Then, $\Hom(G,H) = \oplus_i \Z_{p^{\min(r,b_i)}}$, and thus, $\abs{\Hom(G,H)} \leq |H|$.
\end{fact}
\begin{proof}
Any homomorphism $\phi: \Z_{p^r} \to \Z_{p^b}$, is determined by $\phi(1)$ which must have order $p^{a_\phi} \leq p^{\min(r,b)}$. For a given order $p^j$, the number of elements with order at most $j$ is $p^j$ and thus $\Hom(G,H) = \Z_p^{\min(r,b)}$.  	
\end{proof}

\noindent \textbf{Notation}
Let $f: G\rightarrow H$ be a function. For any ${\vec x} = (x_1,\dots, x_k) \in G^k$, we use the shorthand: $f(\vec x) := (f(x_1),\dots, f(x_k))$.

\section{A General Test}
\label{sec:test}
	For any,
	$\vec x\in G^k$, we define the following evaluation map:
	\[ \Gamma_{\vec x}: \Hom(G,H)\rightarrow H^k~:~\Gamma_{\vec x}(\phi)=\parens{\phi(x_1),\dots,\phi(x_k)} \;. \]
	Because $H$ is an abelian group, the set $\Hom(G,H)$ is an abelian group under pointwise multiplication. Moreover, the maps $\{ \Gamma_{\vec x}\}_{\vec x}$ are homomorphisms. We make the following important definitions that we will use throughout:  
%	denote its (fractional) size in $G^k$ by $\eta_k.$
	\[
\sfH_{\vec x} := \Ima\parens{\Gamma_{\vec x}} \leqslant H^k, \quad	\cG_k := \braces[\big]{\vec x \in G^k \mid \sfH_{\vec x}\neq H^k}, \quad\eta_k = \frac{|\cG_k|}{|G^k|}\;.
	\]
	
	 %For any homomorphism, $\psi$, we define $\tilde{\ker}(\psi):=\tilde{\ker}(\psi) - e$. 
	
	\begin{lemma}[Rewriting the agreement]\label{lem:alpha_k} Let $G,H$ be finite abelian groups, and let $f:G\to H$ be any function. 
	\[
		\sum_{\phi\in \Hom(G,H)}  \agrHom^k ~=~ \eta_k\cdot \Ex{\vec{x} \sim \cG_k}{ \indicator{f(\vec x) \in \sfH_{\vec x}} \abs{\ker(\Gamma_{\vec x})} } + (1-\eta_k)\cdot \frac{\abs{\Hom(G,H)}}{\abs{H^k}}\;.
	\]		
	\end{lemma}
	\begin{proof}
	By expanding the definition, and using the fact that the expectation of independent samples is a product, we get,
		\begin{align*}
		  \agrHom^k ~&=~ \parens[\Big]{\Ex{x\sim G}{\indicator{f(x) = \phi(x)}}}^k\\		  		 ~&=~ \prod_{i=1}^k{\Ex{x_i \sim G}{\indicator{f(x_i) = \phi(x_i)}}}\\
		  ~&=~ {\Ex{\vec{x} \sim G^k}{\indicator{f(\vec x) = \phi(\vec x)}}}.
		\end{align*}
Now, we sum over the homomorphisms to get,
	\begin{align*}
	\sum_\phi	 \agrHom^k ~&=~ \sum_\phi{\Ex{\vec{x} \sim G^k}{\indicator{f(\vec x) = \phi(\vec x)}}}.\\
	~&=~   \Ex{\vec{x} \sim G^k}{ \sum_\phi \indicator{f(\vec x) = \phi(\vec x)} }.	\\
			~&=~   \Ex{\vec{x} \sim G^k}{ \indicator{f(\vec x) \in \sfH_{\vec x}} \abs{\ker(\Gamma_{\vec x})} }.
\end{align*}
The last equality is a consequence of the fact that if an element $y \in \Hx$, then it is the image of exactly $\abs{\ker((\Gamma_{\vec x}))}$ many homomorphisms. Now, if $\Hx = H^k$, then the indicator is always $1$, and we will separate those terms out. Then, we get, 
\begin{align*}
	\sum_\phi	 \agrHom^k	~&=~   \Ex{\vec{x} \sim G^k}{ \indicator{f(\vec x) \in H_x} \abs{\ker(\Gamma_{\vec x})} } \\
	~&=~ \eta_k\cdot \Ex{\vec{x} \sim \cG_k}{ \indicator{f(\vec x) \in H_x} \abs{\ker(\Gamma_{\vec x})} } + (1-\eta_k)\cdot \frac{\abs{\Hom(G,H)}}{\abs{H^k}}. \qedhere
\end{align*}
%$ : =  \Pr{\vec x \sim G^k}{\sfH_{\vec x} \in \cG_k }$.
	\end{proof}

\paragraph{General Test} Motivated by the non-constant term in the expression, we define the distribution on $\cG_k$ which samples  $x \propto |\ker(\Gamma_{\vec x})|$, \ie 
\[
\cD_{\ker}(\vec x) ~:=~ 
 \frac{ |\ker(\Gamma_{\vec x})| }{\sum_{\vec x \in \cG_k} | \ker(\Gamma_{\vec x})|}\;\;.\]
% Define $\gamma_k = \sum_{\vec x \in \cG_k} | \ker(\Gamma_{\vec x})|$, Then, 
%Then, the main term is the test passing possibility of the following test

	\begin{tcolorbox}[colframe=teal, colback=white, title={$\test\_ {\ker}_k (G,H)$}, label=test:generic]
		\label{box:test}
	\begin{itemize}
		\item Sample $\vec x\sim \cD_{\ker}$, \ie $\vec{x}\in \cG_k \propto \abs{\ker(\Gamma_{\vec x})}$.
		\item If $f(\vec x)\in \sfH_{\vec x}$: return $1$; otherwise: return $0$.
	\end{itemize}
%$D_{\ker}(\vec x)=\frac{|\ker(\Gamma_{\vec x})|}{         \sum_{{\vec x \in G^k~:~ {\sf H}_{\vec x}\neq H^k}}   |\ker(\Gamma_{\vec x})| }$
\end{tcolorbox}

%Let $\delta_k$ be the probability that the function $f$ passes $\test_k$. Then,

\begin{corollary}[Test passing Probability]\label{cor:test_prob}
	Let $f: G\to H$ and $\delta_k$ be the probability that $f$ passes the \hyperref[test:generic]{$\test\_ {\ker}_k (G,H)$}. Then,
	\[
	\delta_k ~=~ \frac{\Ex{\vec{x} \sim \cG_k}{ \indicator{f(\vec x) \in H_x} \abs{\ker(\Gamma_{\vec x})} }}{ \Ex{\vec{x} \sim \cG_k}{\abs*{\ker(\Gamma_{\vec x})}}}\;\;.\]
	And therefore,
	\[
	\sum_\phi \agrHom^k ~=~  \delta_k\cdot  \eta_k\cdot \Ex{\vec{x} \sim \cG_k}{\abs*{\ker(\Gamma_{\vec x})}} ~+~  (1-\eta_k)\cdot\frac{\Hom(G,H)}{|H|^k}.
	\]
\end{corollary}

The following claim merely gives an alternate way to compute the term on the RHS of the above equation, \ie the expected kernel size.

\begin{claim}
	\label{clm:gamma1} 
	For any finite groups $G,H$, and $k\geq 1$, we have	
	\[
	\gamma_k ~:=~ \sum_{\vec x \in G^k}{\abs*{\ker(\Gamma_{\vec x})}} ~~= \sum_{\phi \in \Hom(G,H)}{\abs{\ker(\phi)}^k}  . \]
%	. Therefore, \[\gamma_k  = \sum_{\phi\in \Hom(G,H)} \abs{\ker(\phi)}^k - (1-\zeta_k)\cdot \frac{|\hom(G,H)|\cdot |G^k| }{H^k}. \]
\end{claim}
\begin{proof} The proof is a simple use of the definition of $\Gamma_{\vec x}$ and switching the order of summation. 
\begin{align*}
		\sum_{ {\vec x} \in G^k }   |\ker(\Gamma_{\vec x})| 
		&~=~\sum_{ {\vec x} \in G^k } \sum_{\phi} \indi_{\{\phi(x_i)=0~\forall i \}}
		\\&~=~ \sum_{\phi} \sum_{\vec x \in G^k } \indi_{\{\phi(x_i)=0~\forall i \}}
		\\&~=~ \sum_\phi \abs{\ker(\phi)}^k . \qedhere
	\end{align*}
\end{proof}

\paragraph{Summarizing the expressions} We now summarize the expressions we need for ready reference in our proofs.
\begin{align}
\label{eq:max_approx}\max_\phi \agrHom ~&\geq~  \frac{\sum_\phi \agrHom^{k+1}}{\sum_\phi \agrHom^{k}}\\
\label{eq:1}\sum_\phi \agrHom^{k} ~&=~   \Ex{\vec{x} \sim G^k}{ \indicator{f(\vec x) \in \sfH_{\vec x}} \abs{\ker(\Gamma_{\vec x})} }\\
\label{eq:2}~&=~\eta_k\cdot \Ex{\vec{x} \sim \cG_k}{ \indicator{f(\vec x) \in {\sfH}_x} \abs{\ker(\Gamma_{\vec x})} } + (1-\eta_k)\cdot \frac{\abs{\Hom(G,H)}}{\abs{H^k}}\\
\label{eq:3}~&=~\delta_k \cdot  \frac{\gamma_k}{|G|^k} + (1-\eta_k)\cdot \frac{\abs{\Hom(G,H)}}{\abs{H^k}}.
\end{align}

%\subsection{Composite cycles}

%\input{blrkiwi.tex}
%\subsection{Generalizing to $p$-groups}

%Let $G$ be a 

\section{Cyclic Groups}\label{sec:cyclic}
\subsection{Cyclic groups of prime power order to abelian groups of small rank.}
\label{subsec:cyclic_bounded}
We will first start with both the domain being a cyclic group of prime power order. For such groups, the expression from~\cref{cor:test_prob} can be further simplified, as $\eta_k = 1$. 

%which will be useful for us. 

%Let $G= \Z_{p^r} $ and $H=\Z_{p^s}$. 
%\todo{For first few sections, we use general test. Reiterate or reorder sections}
%\subsection{Preparatory lemmas}

\begin{observation}\label{lem:eta_one}
Let $G= \Z_{p^r} $ and $H$ be any abelian group. Then, for any $k \geq 2$, $\eta_k(G,H) = 1$ and thus, for any $f: G\to H$,
\[
\sum_{\phi\in\Hom(G,H)} \agrHom^k ~=~  \delta_k(f)\cdot \frac{\gamma_k }{|G|^k}\;\;.
\]	
\end{observation}
\begin{proof}
From~\cref{fact:cyc_homs}, we have $\abs{\Hom(G,H)} = p^{\min(r,s)} \leq |H| < H^k$, for $k \geq 2$. Therefore, the map $\Gamma_{\vec{x}}$ cannot be surjective for any $\vec{x}\in G^k$, and so, $\cG_k = G^k$ and $\eta_k(G,H) = 1$. Now, one can plug this in~\cref{eq:3}. 	
%Any homomorphism $\phi: \Z_{p^r} \to \Z_{p^s}$, is determined by $\phi(1)$ which must have order $p^{a_\phi} \leq p^{\min(r,s)}$. For a given order $p^i$, the number of elements with order at most $i$ is $p^i$ and thus $\abs{\Hom(G,H)} = p^{\min(r,s)} \leq |H| < |H^k|$ for $k \geq 2$. Therefore, $\cG_k = G^k$ in this case and $\eta_k = 1$.
\end{proof}

%\begin{corollary}[of \cref{lem:alpha_k}] 
%	\[\sum_\phi \alpha_\phi^k =   {\frac{\delta_k\gamma_k}{|G^k|}} +  (1-\eta_k)\cdot \frac{\abs{\Hom(G,H)}}{\abs{H^k}}.\]
%\end{corollary} 
%\begin{proof}
%	We have \[\sum_\phi \alpha_\phi^k = \eta_k \delta_k \parens[\bigg]{\frac{\gamma_k}{|\cG_k|}} +  (1-\eta_k)\cdot \frac{\abs{\Hom(G,H)}}{\abs{H^k}}\]
%\end{proof}
%

\paragraph{Bounding $\gamma_k$} From the expression it is clear that the only quantity we need to analyze is $\gamma_k$ which we will do via~\cref{clm:gamma1}. 

%Note that due to~\cref{fact:decomp}, any homomorphism maps $\Z_{p^r}$ only to a \textit{$p$-group} \ie a group of the form, $H = \oplus_i \Z_{p^{b_i}}$. Since, $\gamma_k$ only concerns homomorphisms, we can assume without loss of generality that  $H$ is a $p$ 
%
%\subsection{Cyclic group to abelian groups of small rank}
%We now generalize the above result beyond cyclic groups to the case when the range has a small rank \ie it is a direct sum of a few cyclic groups. We first generalize~\cref{propn: cyclic_gamma} to the case when $H$ is an arbitrary abelian group. 
%The goal of this section is to compute $\gamma_k(G,H):= \sum_{\phi \in \Hom(G,H)} \abs{\ker(\phi)}^k$ for any pair of abelian groups $G,H$. We have already seen this for the case when both are cyclic groups.

\begin{lemma}[Cyclic to Abelian]
	Let $G = \Z_{p^r}$ and $H$ any abelian group. Then, 
\[
\gamma_k(G,H) ~=~ \abs[\big]{\Hom (\Z_{p^{r}}, H)} + \parens{ 1 - p^{-k}}\sum_{a=1}^{r}  p^{ak} \cdot \abs[\big]{\Hom (\Z_{p^{r-a}}, H)}. \]
%	where $b_i' = \min(b, b_i)$, and $S_a = \{ i \mid b-b_i' \leq a \}$.
%}p^{\sum_i\min(r,b_i)} + (1-p^{-k}) \sum_{a=1}^r \brackets[\Big]{ p^{ak} \cdot p^{\sum_i\min(r-a,b_i)} }
\end{lemma}
\begin{proof}
The kernel of any homomorphism $\phi:G\to H$ is $\Z_{p^a}$ for $0 \leq a \leq r$. We thus only need to count the number of homomorphisms with kernel exactly $\Z_{p^a}$. To start we observe that the following sets are in bijection,
\[
\braces{\phi: \Z_{p^a } \subseteq  \ker(\phi)} ~\simeq~    \Hom (\Z_{p^r}/\Z_{p^a}, H) ~\simeq~    \Hom (\Z_{p^{r-a}}, H) .
\] 

Using this we can count the homomorphisms with kernel exactly $\Z_{p^a}$ by excluding those that have a larger kernel, \ie $\Z_{p^{a+1}}$. Thus, for any $ a < r$: 
\[
\braces{\phi: \Z_{p^a } =  \ker(\phi)} ~\simeq~    \Hom (\Z_{p^{r-a}}, H) \setminus \Hom (\Z_{p^{r-a-1}}, H) .
\] 
Using the above bijection, we can perform our computation quite easily as follows,
\begin{align*}
\sum_{\phi\in\Hom(G,H)}|\ker(\phi)|^k ~&=~ \sum_{a=0}^r \sum_{\ker{\phi} = \Z_{p^a}} p^{ak} \\
	~&=~ \sum_{a=0}^{r}  p^{ak} \cdot \abs[\big]{\braces{\phi: \Z_{p^a } =  \ker(\phi)}} \\
		~&=~ p^{rk} + \sum_{a=0}^{r-1}  p^{ak} \cdot \brackets[\big]{\abs[\big]{\Hom (\Z_{p^{r-a}}, H)} -  \abs[\big]{\Hom (\Z_{p^{r-a-1}}, H)}}.
		\end{align*}
We can now rearrange the terms on the right-hand side, to obtain:
\begin{align*}
\gamma_k ~&=~ p^{rk} + \sum_{a=1}^{r-1}  \parens{p^{ak} - p^{(a-1)k}} \cdot \abs[\big]{\Hom (\Z_{p^{r-a}}, H)} + \abs[\big]{\Hom (\Z_{p^{r}}, H)} - p^{(r-1)k}\abs[\big]{\Hom (\Z_{p^{0}}, H)}\\
~&=~  \abs[\big]{\Hom (\Z_{p^{r}}, H)} + \sum_{a=1}^{r}  \parens{p^{ak} - p^{(a-1)k}} \cdot \abs[\big]{\Hom (\Z_{p^{r-a}}, H)}\\
~&=~  \abs[\big]{\Hom (\Z_{p^{r}}, H)} + \sum_{a=1}^{r}  p^{ak}\parens{ 1 - p^{-k}} \cdot \abs[\big]{\Hom (\Z_{p^{r-a}}, H)}. \qedhere
%~&=~  \abs[\big]{\Hom (\Z_{p^{r}}, H)} + \sum_{a=1}^{r}  p^{ak}\parens{ 1 - p^{-k}} \cdot p^{\sum_i \min (r-a,b_i)}. \qedhere
\end{align*}
\end{proof}

Note that due to~\cref{fact:decomp}, any homomorphism maps $\Z_{p^r}$ only to a \textit{$p$-group}. Since, $\gamma_k$ only concerns homomorphisms, it is only dependent on the $p$-component $H_p := \oplus_{i=1}^t \Z_{p^{b_i}}$. Here, $t$ is the $p$-rank of $H$. We now explicitly bound $\gamma_k$ for $k$ larger that the $p$-rank of $H$.

%\todo{Move def of p rank somewhere?}
%
% \ie $\gamma_k(G,H) =  $is a $p$ 

\begin{corollary}\label{cor:boundedrank}
Let $G = \Z_{p^r}$, and let $H$ be an abelian group of $p$-rank $t$, and let $k >  t$. Then, 
	\[
(1- p^{-k})\cdot p^{kr}   ~\leq~	\gamma_k(G,H) ~\leq~  \parens[\Big]{\frac{p^{k-t}}{p^{k-t}-1}}  \cdot  p^{kr} 
	.\]
%	Therefore, $\frac{\gamma_{k+1}}{\gamma_k} ~\geq~ \frac{(p-1)(1- p^{-k})}{p} \cdot  p^r  ~\geq~  \parens[\big]{\frac{p-1}{p}}^2 |G|$.
\end{corollary}
\begin{proof}
The lower bound is directly obtained from the expression by only picking the term corresponding to $a =0$.
	From, the expression in 
	\[ \parens{ 1 - p^{-k}} \sum_{a=0}^{r}  p^{ak} \cdot \abs[\big]{\Hom (\Z_{p^{r-a}}, H)}  ~\leq~ \gamma_k(G,H) ~\leq~ \sum_{a=0}^{r}  p^{ak} \cdot \abs[\big]{\Hom (\Z_{p^{r-a}}, H)}. \]

Now, for any $H$ of $p$-rank $t$, we have from~\cref{fact:cyc_homs}:
% let it $p$-component be $H_p = $
\[ \abs[\big]{\Hom (\Z_{p^{r-a}}, H)} ~=~ \prod_{i=1}^t p^{\min(r-a,b_i)} ~\leq~ p^{(r-a)t}.\]
The upper bound then can be calculated as:
\begin{align*}
\gamma_k ~&\leq~ \sum_{a=0}^r p^{ak}p^{(r-a)t} ~=~ p^{rt}\sum_{a=0}^r p^{(k-t)a}\\
	 ~&\leq~ p^{rt}p^{(k-t)r} \parens[\Big]{1+\frac{1}{p^{k-t}-1}}. \qedhere 
\end{align*}
%The final conclusion follows directly from the above upper and lower bounds.
\end{proof}

We can now use the above calculation for $\gamma_k$ to deduce a testing result, and a (combinatorial) list decoding bound.

\begin{theorem}[Testing prime power cyclic groups to Abelian groups of bounded rank]\label{thm:cyclic_bounded}
		Let  $G=\Z_{p^r}$ be a cyclic group and $H$ be an abelian group of $p$-rank $t\geq 1$. Let $k \geq t+2$ be an integer, and let $f: G\to H$ be any function. Then if $f$ passes 
		\hyperref[test:generic]{$\test\_\ker_k(G,H)$} 
		with probability $\delta_k$, then,
		 \[
	  \parens[\Big]{\frac{(p-1)^2}{p^2}\cdot \delta_k}^{\frac{1}{k-t-1}} ~\leq~ \agr(f,\phi)  ~\leq~ \parens[\Big]{\frac{p^2}{p^2-1}\delta_k}^{\frac{1}{k}} 	 \;  . \] 
\end{theorem}
\begin{proof}
The upper bound directly follows from~\cref{lem:eta_one} and \cref{cor:boundedrank}.  For the lower bound we use~\cref{eq:max_approx} and \cref{lem:eta_one} to get,
\begin{align*}
\max_{\phi}\, \agr(f,\phi) ~&\geq~  \frac{\sum_{\phi}{\agr(f,\phi)^i}}{\sum_{\phi}{\agr(f,\phi)^{i-1}}} ~\geq~  \frac{\delta_i \gamma_i}{|G| \delta_{i-1} \gamma_{i-1}} .	
\end{align*}

Multiplying this for $i \in [t+2, k]$, we get
\begin{align*}
\parens{\max_{\phi} \agr(f,\phi)}^{k-t-1} ~&\geq~ \parens[\Big]{\frac{1}{|G|}}^{k-t} \cdot \frac{\delta_k \gamma_k}{\delta_{t+1} \gamma_{t+1}}\\
~&\geq~ \parens[\Big]{\frac{1}{|G|}}^{k-t-1} \cdot \frac{\delta_k \gamma_k}{\gamma_{t+1}}  && [\delta_{t+1} \leq 1]\\
~&\geq~  \delta_k  \cdot \frac{(1-p^{-k})(p-1)}{p}  && [\text{Using~\cref{cor:boundedrank}}]\\
~&\geq~  \delta_k  \cdot \frac{(p-1)^2}{ p^2 }  \qedhere
\end{align*} 
\end{proof}

Using the above bound we also immediately get a list decoding bound.

\begin{theorem}[List Size Bound]\label{theo:list_size}
	Let  $G=\Z_{p^r}$ be a cyclic group and $H$ be an abelian group of $p$-rank $t\geq 1$. Let $f: G\to H$ be any function.  Then the following holds:
		\[ \abs{\{ \phi\in \Hom(G,H)  ~:~\agr(f,\phi) ~\geq~ \ep  \}} ~\leq~ \parens[\Big]{\frac{p}{p-1}} \cdot  \frac{1}{\ep^{t+1}}.\]
		In particular, we get a list size bound of $\frac{2}{\ep^2}$ for homomorphisms between cyclic groups.
\end{theorem}
\begin{proof}
Let $N = \abs{\{ \phi\in \hom(G,H)  ~:~\agr(f,\phi) ~\geq~ \ep  \}}$.
Then,
\begin{align*}
N \ep^{t+1} ~\leq~	\sum_\phi \alpha_\phi^{t+1} ~&=~ \frac{\delta_{t+1}\gamma_{t+1}}{|G|^{t+1}} \\ 
 ~&\leq~ \frac{p}{p-1} && [\text{Using~\cref{cor:boundedrank}}].\qedhere 
\end{align*}
% be the number of homomo
%	We have: $\sum_\phi \alpha_\phi^2 = \frac{\delta_2 \gamma_{2}}{|G|^2}$. From~\cref{propn: cyclic_gamma}, it holds that $\gamma_2\leq |G|^2\cdot \frac{2(p-1)}{p}.$ As $\delta_2\leq 1$, it follows that: $\sum_\phi \alpha_\phi^2 \leq 2.$ The claim follows. 
\end{proof}

%\begin{theorem}[List Size Bound]
%		Let G = $\Z_{p^r}, H= \Z_{p^s}$ and $f: G\to H$ be any function.  Then the following holds:
%		\[ \abs{\{ \phi\in \hom(G,H)  ~:~\alpha_{\phi} \geq \ep  \}} ~\leq~ \frac{2}{\ep^2}\]
%\end{theorem}
%\begin{proof}
%	We have: $\sum_\phi \alpha_\phi^2 = \frac{\delta_2 \gamma_{2}}{|G|^2}$. From~\cref{propn: cyclic_gamma}, it holds that $\gamma_2\leq |G|^2\cdot \frac{2(p-1)}{p}.$ As $\delta_2\leq 1$, it follows that: $\sum_\phi \alpha_\phi^2 \leq 2.$ The claim follows. 
%\end{proof}

\subsection{Arbitrary Cyclic groups}
\label{subsec:cyclic_gen}

To handle general cyclic groups, we will use the decomposition of abelian groups into their $p$-components. 

\begin{lemma}[Reduction to p-groups]\label{lem:reduction_p}
	Let $\phi: G\to H$, and let $X = \oplus_p X_p$ for $X\in \{G,H\}$ be the decomposition of the groups into their p-components. Then, $\phi = \oplus \phi_p$ where $\phi_p: G_p\to H_p$. Therefore,
	\begin{align*}
		\gamma_k(G,H) ~&=~ \prod_p \gamma_k(G_p, H_p),\\
		1-\eta_k(G,H) ~&=~ \prod_p \parens{1-\eta_k(G_p, H_p)}.
	\end{align*} 
\end{lemma}
\begin{proof}
Any homomorphism is a homomorphism between the respective $p$-groups, \ie $\Hom(G,H) = \oplus_p \Hom(G_p, H_p)$.	By~\cref{clm:gamma1}, 
\begin{align*}
\gamma_k(G,H) ~&=~ \sum_{\phi \in \Hom(G,H)} \abs{\ker{\phi}}^k \;,\\
~&=~  \sum_{\phi = (\phi_p)_p} \prod_p \abs{\ker{\phi_p}}^k \;,\\
~&=~  \prod_p\sum_{\phi_p \in \Hom(G_p,H_p)}  \abs{\ker{\phi_p}}^k \;,\\
~&=~ \prod_p \gamma_k(G_p,H_p).
\end{align*}
Since, $\Hom(G,H) = \oplus_p \Hom(G_p, H_p)$, for any $\vec{x} \in G^k$, we can decompose the map $\Gamma_{\vec{x}} : \Hom(G,H) \to H^k$ as direct sum, $\Gamma_{\vec{x}} = \oplus_p \Gamma_{\vec{x}_p}^{(p)}$, where $\Gamma_{\vec{x}_p}^{(p)} : \Hom(G_p,H_p) \to H_p^k$. 
\begin{align*}
1-\eta_k(G,H) ~&=~ \frac{\abs{\vec{x} \in G^k \mid \mathrm{Im}(\Gamma_{\vec{x}}) = H^k }}{|G|^k} \;,\\
~&=~ \prod_p \frac{\abs{\vec{x}_p \in G_p^k \mid \mathrm{Im}(\Gamma_{\vec{x}_p}^{(p)}) = H_p^k }}{|G|^k} \;,\\
~&=~  \prod_p \parens{1-\eta_k(G_p, H_p)}. \qedhere
\end{align*}
\end{proof}

\begin{definition}[Riemann Zeta Function]\label{def:riem_zeta}
Define the Riemann zeta function using Euler's product formula as $\zeta(s) = \prod_p(1-\frac{1}{p^{s}})^{-1}$ where the product runs over all primes. 	
\end{definition}

\begin{proposition}\label{propn: cyclic_gamma_gen}
	Let $G,H$ be any cyclic groups and let $k \geq 3$. Denote by $\zeta()$, the Riemann zeta function. Then, 
%	\[\frac{\gamma_{k+1}(G,H)}{\gamma_k(G,H)} ~\geq~ |G|\cdot \frac{1}{\zeta(k-1)^2} .\]
		\[ |G|^k ~\leq~ \gamma_k(G,H) ~\leq~ |G|^k \cdot \zeta(k-1)^2 .\]
%	Therefore, for any $k \geq 2$,  
%	\[
%\frac{\gamma_{k+1}}{\gamma_k} ~\geq~ |G|\cdot \parens[\big]{1+ 2p^{-(k-1)}}^{-1}.
%\]
\end{proposition}
\begin{proof}
Recall that for any cyclic group, their $p$-components are cyclic, \ie we have that for $G = \oplus_p G_p, H = \oplus_p H_p $, each of $G_p, H_p$ are cyclic (potentially trivial). Moreover, $G_p$ is non-trivial if and only if $p \mid |G|$. Therefore, we can use~\cref{cor:boundedrank} (for $t=1$) in conjunction with~\cref{lem:reduction_p} to get,
	\begin{align*}
		\gamma_k(G,H) ~&=~  \prod_p \gamma_k(G_p,H_p) \;, \\
		 ~&\leq~  \prod_{p\mid |G|} |G_p|^k \parens[\Big]{1+\frac{2}{p^{k-1}}} \;,\\
		 ~&\leq~  \prod_{p\mid |G|} |G_p|^k \parens[\Big]{1-\frac{1}{p^{k-1}}}^{-2} \;,\\
		 ~&=~  |G|^k \prod_{p\mid |G|}  \parens[\Big]{1-\frac{1}{p^{k-1}}}^{-2} \;,\\
		 ~&\leq~  |G|^k \prod_{p}  \parens[\Big]{1-\frac{1}{p^{k-1}}}^{-2} ~=~ {|G|^k}\cdot {\zeta(k-1)^2}. 
	\end{align*}
	
%	\begin{align*}
%		\frac{\gamma_{k+1}(G,H)}{\gamma_k(G,H)} ~&=~  \prod_p \frac{\gamma_{k+1}(G_p,H_p)}{\gamma_k(G_p,H_p)} \\
%		 ~&\geq~  \prod_{p\mid |G|} |G_p| \parens[\Big]{1+\frac{2}{p^{k-1}}}^{-1}\\
%		 ~&\geq~  \prod_{p\mid |G|} |G_p| \parens[\Big]{1-\frac{1}{p^{k-1}}}^2\\
%		 ~&=~  |G| \prod_{p\mid |G|}  \parens[\Big]{1-\frac{1}{p^{k-1}}}^2\\
%		 ~&\geq~  |G| \prod_{p}  \parens[\Big]{1-\frac{1}{p^{k-1}}}^2 ~=~ \frac{|G|}{\zeta(k-1)^2}. 
%	\end{align*}
The last line gives an upper bound by taking a product over all primes and using Euler's product formula~\cref{def:riem_zeta}. The lower bound for $\gamma_k$ directly follows from the lower bounds for $\gamma_k(G_p, H_p)$.
\end{proof}

\begin{remark}
The above expression can be analyzed more carefully and perhaps one can obtain bounds for $k=2$	which would yield a $3$-query test. We instead opt to keep the presentation clean at the cost of converting the test to a $k \geq 4$ query test. We now analyze the guarantee of the test just as in~\cref{thm:cyclic_bounded}.
\end{remark}

%\begin{fact}
%	For any $n$, $\prod_{p \mid n}(1-\frac{1}{p^{k-1}}) ~\geq~ \prod_p(1-\frac{1}{p^{s}})^{-1} ~=~ \frac{1}{\zeta(k-1)} $
%\end{fact}
%
%\begin{lemma}
%Let $G = \Z_m$, if $m = \prod_p p^{a_p}$ $G = \oplus_p G_p = \oplus_p \Z_{p^{a_p}}$.
%\end{lemma}
%\begin{proof}
%	Let $G = \Z_m$, if $m = \prod_p p^{a_p}$ $G = \oplus_p G_p = \oplus_p \Z_{p^{a_p}}$.
%\end{proof}
%\begin{theorem}[Testing]
%		Let $G, H$ be any cyclic groups and $f: G\to H$ be any function. Let $k\geq 4$ be any integer. Then if $f$ passes $\test_k$ with probability $\delta_k$, then, $d(f, \Hom(G,H)) \leq \delta_k^{\frac{1}{k-2}}$. 
%\end{theorem}
%\begin{proof}
%	\todo{Add proof.}
%	
%\end{proof}

\begin{theorem}[Testing $\Hom(\Z_n, \Z_m)$]\label{thm:cyclic_gen}
		Let $G, H$ be any cyclic groups and $f: G\to H$ be any function. Let $k\geq 4$ be any integer. Then if $f$ passes 	\hyperref[test:generic]{$\test\_\ker_k(G,H)$}  with probability $\delta_k$, then there exists a homomorphism $\phi \in \Hom(G,H)$ such that $\agr(f,\phi) ~\geq~  \parens[\big]{\zeta(2)^2\cdot\delta_k}^{\frac{1}{k-3}} $ .
\end{theorem}
\begin{proof}
By \cref{lem:reduction_p}, $\eta(G,H) =1$ for any pair of cyclic groups, and thus, the expression from~\cref{lem:eta_one} holds. Using it, 
\begin{align*}
\max_{\phi}\, \agr(f,\phi) ~&\geq~  \frac{\sum_{\phi}{\agr(f,\phi)^i}}{\sum_{\phi}{\agr(f,\phi)^{i-1}}} ~\geq~  \frac{\delta_i \gamma_i}{|G| \delta_{i-1} \gamma_{i-1}} .	
\end{align*}

%\[
%\max_{\phi}\, \agr(f,\phi) ~\geq~ \frac{\delta_i(f)\cdot \gamma_i}{|G| \delta_{i-1}(f)\cdot \gamma_{i-1}} 
%\]
Multiplying this for $i \in [4, k]$, we get
\begin{align*}
\parens{\max_{\phi} \agr(f,\phi)}^{k-t-1} ~&\geq~ \parens[\Big]{\frac{1}{|G|}}^{k-t} \cdot \frac{\delta_k \gamma_k}{\delta_{t+1} \gamma_{t+1}}\\
~&\geq~ \parens[\Big]{\frac{1}{|G|}}^{k-t-1} \cdot \frac{\delta_k \cdot  \gamma_k}{\gamma_{t+1}}  && [\delta_{t+1} \leq 1]\\
~&\geq~  \delta_k  \cdot \frac{(1-p^{-k})(p-1)}{p}  && [\text{Using~\cref{cor:boundedrank}}]\\
~&\geq~  \delta_k  \cdot \frac{(p-1)^2}{ p^2 } \;. \qedhere
\end{align*} 
%For any $i \geq 3$, \todo{cite relevant lemma}
%\[
%\max_{\phi}\, \agr(f,\phi) ~\geq~ \frac{\delta_{i+1} \gamma_{i+1}}{|G| \delta_{i} \gamma_{i}} \]
%Multiplying this for $i \in [3, k-1]$, we get
%\[
%\parens{\max_{\phi} \agr(f,\phi)}^{k-2} ~\geq~ \parens[\Big]{\frac{1}{|G|}}^{k-2} \cdot \frac{\delta_k\gamma_{k}}{\delta_{3}\gamma_3}  ~\geq~ \parens[\Big]{\frac{p-1}{p}}^{2(k-t)} \delta_k. 
% \]	
\end{proof}

\section{Vector space over finite fields}
\label{sec:vspace}
Let, $\F_{q}$ be a finite field of order $q$.  In this section, we will focus on functions between a $\F_{q}$-vector space and the finite field, $\F_q$.

\subsection{Vector space to Finite Field}
Let $G=\F^n_q$ and $H=\F_{q}$. Let $f: G \to H$ be an arbitrary function. Recall the expression for $\sum_{\phi} \agrHom^k$ from~\cref{sec:test}:
	\[
\sum_{\phi \in \Hom(G,H)} \agrHom^k ~=~  \delta_k\cdot  \eta_k\cdot \Ex{\vec{x} \sim \cG_k}{\abs*{\ker(\Gamma_{\vec x})}} ~+~  (1-\eta_k)\cdot\frac{\Hom(G,H)}{|H|^k},
\]
where $\delta_k$ is the test passing probability 
of the general test (that samples $\vec x \propto |\ker \Gamma_{\vec x} | $ and checks $\vtest$) given in~\cref{box:test}. A key difference in the vector space case compared to the cyclic case is that $\eta_k \approx 0$. Therefore most of the contribution 
in $\sum_{\phi \in \Hom(G,H)} \agrHom^k$ comes from the non-test part: $\Hom(G,H)\cdot H^{-k}$. This happens because any function has roughly $\frac{1}{q}$-agreement with many homomorphisms (at least $\frac{1}{2q}$-fraction of all homomorphisms). In particular, it is easy to see that if $f$ is a random function then it has $\frac{1}{q}$-agreement with almost all the homomorphisms. This suggests adjusting the definition of agreement, $\agrHom$, so that it measures the non-trivial advantage over the trivial agreement of $\frac{1}{q}$. To achieve this we define the following shifted variant of $\agrHom$.
\[ \textbf{Shifted agreement: }~~\widetilde{\agr}(f,\phi)=\frac{q\agrHom-1}{q-1}\]
Accordingly, we will use the expression $\sum_{\phi} \widetilde{\agr}(f,\phi)^k$ instead of $\sum_{\phi} \agrHom^k$ so that distribution concentrates on homomorphisms with non-trivial agreements. Fortunately, the former expression can be directly computed from the latter via binomial expansion. 

\paragraph{Refining the test.} We need to address one more issue in this vector space setting. 
 We cannot directly use the original generalized test that - samples $\vec x$ with probability  $\propto |\ker \Gamma_{\vec x} |$ and checks $f(\vec x)\in \Hx$. This is because even though this test uses length $k$-tuple, it can happen that sampled the tuple satisfy only linear relation involving $j$ (for some $j<k$) co-ordinates of the input and all the rest of the co-ordinates are independent.   If this happens, then the test essentially collapses to a $j$-th level test involving $j$-length tuple - as it essentially checks if $f$ satisfy that linear relation or not. In other words, the generalized test is a linear combination of such different tests associated with different levels. For a precise analysis, it is crucial to refine the initial test definition and isolate these different tests. To do this, we need to formalize this notion of \emph{tuples satisfiying a linear relation involving $j$-coordinates}. 
 
%  We do this in the following paragraph. 
 
 \begin{definition}[A level-$j$ distribution]\label{def:rj}
 	We define ${\cal R}_j$ to be the set of tuples $\vec x \in G^k$ such that (1) there exist  $\emptyset\neq  S\subseteq [k]$  with $|S|=j$ such that $ \sum_{\ell\in S} a_\ell x_\ell=0$ for non-zero coefficients $\{a_\ell\}_{\ell\in S}$ and (2) vectors in $\vec x$ do not satisfy any other linear relation. 
 \end{definition}
%\paragraph{A level-$j$ distribution} 
 Observe that the sets ${\cal R}_j$ for distinct $j$ are disjoint. Now we collect all the claims involving ${\cal R}_j$ and linear independence of tuples, that will be needed for our analysis.  For any two quantities, ${\sf t}_1$ and ${\sf t}_2$ : we write ${\sf t}_1\approx_{\ep}  {\sf t}_2 $ if 
$\abs{{\sf t}_1-{\sf t}_2}=O(\ep).$ 
Our first claim is the following:

%$ \sum_{\ell\in S} a_\ell x_\ell=0$ for non-zero coefficients $\{a_\ell\}_{\ell\in S}$
%Note that the uniform distribution, ${\cal U}_{{ \cal R}_j}$ on ${\cal R}_j$ is equivalent to the following procedure. 
%\begin{mdframed}[linewidth=0.8pt]
%	\textbf{ Uniform distribution, ${\cal U}_{{ \cal R}_j}$ on ${\cal R}_j$ :}
%\begin{enumerate}
%	\item Draw $y_{1},\dots, y_{{j-1}}\sim G^{j-1}$ such that $\rank(y_{1},\dots, y_{{j-1}})=j-1$. 
	
%	\item Draw $a_{1},\dots, a_{i_{j-1}}\sim \F_{q}\setminus 0$.
%	\item Set $y_{j}=\sum^{j-1}_{l=1}a_ly_l $
%	\item Sample $y_{j+1},\dots,y_k$ such that $\rank(\vec y)=k-1$.
%	\item Sample a random permutation $\pi \sim S_{k}$ and output $\vec x=(y_{\pi})$
%\end{enumerate}
%\end{mdframed}

%We will need the following easy claim. 

\begin{claim}\label{claim:tupleProb}
	Let, $k\geq 1$ be any integer such that $k=O_{n}(1)$. Then it holds that, 
	\begin{enumerate}
		%\item $\# \text{ of }\vec x \in G^k$
		\item $  \Pr{\vec x\sim G^k}{\rank(\vec x)=k} \approx_{q^{-2n}}   \parens [\Big]{1- \frac{q^k -1}{q^n(q-1)}} $.
		\item $\Pr{\vec x\in G^k}{ \vec x \in {\cal R}_{j} }   \approx_{q^{-2n}}  \binom{k}{j} \frac{(q-1)^{j-1}}{q^n}   $.
	\end{enumerate}
	
\end{claim}

\begin{proof}
	For the first claim, we have
	\begin{align*}
	 \Pr{\vec x\sim G^k}{\rank(\vec x)=k}&~=~\frac{	\text{\# of rank $k$ tuples }}{q^{kn}}\\&~=~\frac{1}{q^{kn}} \cdot\prod^{k-1}_{j=0}(q^n-q^j)\\&~=~ \frac{1}{q^{kn}} \cdot \parens[\Big] {q^{kn}-q^{(k-1)n}\cdot \sum^{k-1}_{\ell=0}q^{\ell}  +cq^{(k-2)n}} \text{ for some } c=O_k(1) \;.
	\end{align*}
	%The claim follows by dividing the above by $q^{kn}$, i.e., total number of $k$-length tuples. 
	%\[ \text{\# of rank $i$ tuples }=\prod^{i-1}_{j=0}(q^n-q^j)= q^{in}-q^{(i-1)n}\cdot \sum^{i-1}_{l=0}q^{l}  +cq^{(i-2)n} \text{ for some } c=O_k(1) \]
	
	For the second claim, observe that: given a nonempty $S\subseteq [k]$ and coefficients
	$\{a_\ell\}_{\ell\in S}$, there are $\prod^{j-2}_{t=0}(q^n-q^t)$ number of tuples that satisfy the relation defined by $(S, \{a_\ell\}_{\ell\in S} )$ and no other relations. For $S$, we have $\binom{k}{j}$ choices and for each such $S$, the number of distinct choices for the non-zero coefficients  $\{a_\ell\}_{\ell\in S}$  is $(q-1)^n$. Thus, 
	\begin{equation}
		|{\cal R}_j| ~=~ \binom{k}{j}\cdot (q-1)^n\cdot \prod^{j-2}_{t=0}(q^n-q^t)  \;. \label{eqn:Rj}
	\end{equation}
	Using a similar approximation (ignoring lower order terms) for~\cref{eqn:Rj} as in the first claim, the second claim also follows. 
\end{proof}

\begin{claim}\label{claim:deltaj}
	Let, $1\leq j\leq k$ be two integers. Then, \[ \Ex{\vec x\sim G^k}{\indicator{f(\vec x) \in \Hx} | \vec x\in {\cal R}_j } ~=~ \Ex{\vec x\sim G^{j}}{\indicator{f(\vec x) \in \Hx}  | \vec x\in {\cal R}_j  } \;.
	\]
\end{claim}

\begin{proof}
	%Let $S\subset [k]$ be any subset of size $j$. 
	For any tuple $\vec x \in {\cal R}_j$ and any non-empty subset $S\subseteq [k]$ we say $\vec x $ satisfy $S$ if  $ \sum_{\ell\in S} a_\ell x_\ell=0$ for some non-zero coeffcients $\{a_\ell\}_{\ell\in S}$. 
	\begin{align*}
		&~\Ex{\vec x\sim G^k}{\vtest | \vec x\in {\cal R}_j } 
		\\&~=~\Ex{S=\{i_1< \dots< i_j\}}{  \Ex{\vec x \text{ satisfy } S}{  \vtest | \vec x\in {\cal R}_j   }                }
		\\&~=~
		\Ex{S=\{i_1< \dots< i_j\}}{  \Ex{\vec x \text{ satisfy } S}{  \indicator{f(x_{i_1},\dots, x_{i_j})\in \sfH_{ x_{i_1},\dots, x_{i_j}})} }              }
		\\&~=~\Ex{(y_1,\dots,y_j) }{ \Ex{i_1<\dots<i_j}{ \indicator{f(x_{i_1},\dots, x_{i_j})\in \sfH_{ x_{i_1},\dots, x_{i_j}})} }~ \Big |~ \vec y \in {\cal R}_j } && \small{[\text{ By setting $x_{i_1}=y_1,\dots,x_{i_j}=y_j$}]}
		\\&~=~
		\Ex{(y_1,\dots,y_j) }{  \indicator{f(\vec y)\in \sfH_{ \vec y})}   |  \vec y \in {\cal R}_j }.
	\end{align*}
	Here, the second equality follows because if $\vec x\in {\cal R}_j$, the co-ordinates outside the sampled index subset $S$ are independent and have no impact on test passing. 
\end{proof}

\paragraph{Analyzing the key expression.} 
Let us now examine the expression for $\sum_{\phi} \agrHom^k$, from which the appropriate definition of the test becomes apparent. 

\begin{claim}\label{claim: Vexpression}
	 Let, $k\geq 1$ be any integer, and $f:\F_q^n\to \F_q$ be any function. Then, 
	 \[\sum_{\phi} \agrHom^k ~\approx_{q^{-n}}~ \Pr{\vec x\sim G^k}{\rank(\vec x)=k}  \cdot q^{n-k} + q^{-(k-1)} \cdot    \sum^{k}_{j=1} \binom{k}{j} (q-1)^{j-1} \delta_j(f)\;, \]
	 	where $\delta_j(f) ~:=~  \Ex{\vec x\sim G^j}{\vtest|\vec x\in {\cal R}_j}.$
\end{claim}

\begin{proof}
	For any $\vec x\in G^k$, define 
	\[ \beta(\vec x) ~:=~ \indicator{f(\vec x) \in \Hx} \cdot \abs{\ker(\Gamma_{\vec x}) } ~=~ \indicator{f(\vec x) \in \Hx}q^{n-\rank(\vec x)}. 
	\]
%	As we have, $|\ker(\Gamma_{\vec x})|=|\Hom(\F^n_{q},\F_{q})|\cdot q^{-\rank(\vec x)}$, it holds that $\beta(\vec x)=\indicator{f(\vec x) \in \Hx} \cdot q^{n-\rank(\vec x)}$. 
Also define, 
	\[ \sfT_{k} ~:=~ \Pr{\vec x\in G^k}{ \rank(\vec x)=k }\cdot  \Ex{\vec x\in G^k}{ \beta(\vec x) \mid  \rank(\vec x)=k}.\] 
	Similarly we define $\sfT_{k-1}$ and $\sfT_{\leq k-2}$ with respect to the event rank being $k-1$ and at most $k-2$ respectively.  It follows that: 
	\begin{align*}
		\sum_{\phi } \agrHom^k~= \Ex{\vec x\sim G^k} {\beta(\vec x)}~=~ \sfT_{\leq k-2} +  \sfT_{k-1} +  \sfT_{k} \quad .
	\end{align*}
	%\[ \]
	Computing  $ \sfT_{k}$ is straightforward. This is because if $\rank(\vec x)=k$, then $\sfH_{\vec x}=H^k$; consequently, $\indicator{f(\vec x) \in \Hx}=1$ for all such full rank $\vec x$. 
	So, $ \sfT_{k}=   \Pr{\vec x\sim G^k}{\rank(\vec x)=k} \cdot q^{n-k}.$ Using~\cref{claim:tupleProb}, we have 
	\[  \sfT_{k}  ~\approx_{q^{-n}}~ \parens[\Big]{1- \frac{q^k -1}{q^n(q-1)}    } \cdot q^{n-k}  \quad .\]
	
	To compute $ \sfT_{k-1}$, we further partition the set of rank $k-1$ tuples by the linear relation they satisfy. Observe that any such tuple, $\vec x$, must satisfy exactly one relation. Thus each such tuple belongs to ${\cal R}_j$ for some $j\in \{1,\dots k \}$ and this is the partitioning we use.
	%Moroever, the relation must be of the form
	%$\sum_{i\in S}c_ix_i$ for some $\emptyset \neq S\subseteq [i]$ and  $c_i\in \F^*_{q}$ for $i\in S$ with $c_{\max\{S\}}=1$. 
	\begin{align*}
		\sfT_{k-1}&~=~\Pr{\vec x\in G^k}{ \rank(\vec x)=k-1 }\cdot  \Ex{\vec x\in G^k}{ \beta(\vec x) \mid \rank(\vec x)=k-1 }\; ,\\&~=~
		\sum^{k}_{j=1}   \Pr{\vec x\in G^k}{ \vec x \in {\cal R}_{j} }
		\cdot \Ex{\vec x\in G^k}{ \beta(\vec x) \mid  \vec x \in {\cal R}_{j} } \; ,
		\\&~=~   q^{n-(k-1)} \cdot    \sum^{k}_{j=1}   \Pr{\vec x\in G^k}{ \vec x \in {\cal R}_{j} }  \cdot  
		 \Ex{\vec x\in G^k}{ \vtest \mid  \vec x \in {\cal R}_{j} } \; ,
		\\&~=~  q^{n-(k-1)} \cdot    \sum^{k}_{j=1} \binom{k}{j} \frac{(q-1)^{j-1}}{q^n} \cdot  
		\Ex{\vec x\in G^k}{ \vtest \mid  \vec x \in {\cal R}_{j} }\; , && [\text{By \cref{claim:tupleProb}}]
		% \sum_{\vec x~:~\rank(\vec x)=k-1}   \indicator{f(\vec x) \in \Hx} \cdot q^{n-(k-1)}
		% \\&~=~q^{n-(k-1)} \cdot \sum^{k}_{j=1} \sum_{\vec x \in {\cal R_j}  } \indicator{f(\vec x) \in \Hx} 
		%\\&~=~ q^{n-(k-1)}  \sum^{k}_{j=1} |{\cal R}_j|  \cdot \delta_j 
			\\&~=~  q^{n-(k-1)} \cdot    \sum^{k}_{j=1} \binom{k}{j} \frac{(q-1)^{j-1}}{q^n} \cdot  
		\Ex{\vec x\in G^j}{ \vtest \mid  \vec x \in {\cal R}_{j} } \; .&& [\text{By \cref{claim:deltaj}}]
	\end{align*}
	As in \cref{claim:tupleProb}, using straightforward counting, one can show $T_{k-2}=O(q^{-2n})$. Therefore,

	\begin{equation}\label{eq:t1t2}
		\sum_{\phi } \agrHom^k ~\approx_{q^{-n}}~   \sfT_{k-1} +  \sfT_{k}\; .
	\end{equation}

	The claim follows. 
\end{proof}

\paragraph{Defining the refined test. } Inspired from~\cref{claim: Vexpression} we define a $k^{\mathrm{th}}$ test such that the test passing probability of a function $f$ is $\delta_k(f)$ as above, \ie  
\[\delta_k(f) :=	\Ex{\vec x\in G^k}{ \vtest |  \vec x \in {\cal R}_{k} }.\]

% We exactly do that in the following refined $k$-th level test. 

%with a bit of cleanup to make sure the cancellations go smoothly. 

%Towards this, the next short claim shows that the quantity $\Ex{\vec x\in G^k}{ \vtest |  \vec x \in {\cal R}_{j} } $ only depends on $j$. 

%The key claim that we will prove is the following: 
%\[ \sum_{\phi}  \tilde{\alpha}_\phi^k =  \frac{1}{q-1}(q\delta_k- 1) \]

%\cref{claim:deltaj}, inspires the following refined $k$-th level test: 
\begin{tcolorbox}[colframe=teal, colback=white, title=${\test\_\mathsf{VSpace}}_k(f)$, label=test:vspace]
	%\label{test:vspace}
	\begin{itemize}
		\item Sample $(x_1,\dots, x_k) \sim {\cal R}_k$.
		\item If $(f(x_1),\dots, f(x_k))\in \sfH_{\vec x}$: return $1$; otherwise: return $0$
	\end{itemize}
	%$D_{\ker}(\vec x)=\frac{|\ker(\Gamma_{\vec x})|}{         \sum_{{\vec x \in G^k~:~ {\sf H}_{\vec x}\neq H^k}}   |\ker(\Gamma_{\vec x})| }$
Or equivalently, 
	\begin{itemize}
	\item Sample $k-1$ independent vectors:  $x_1,\dots, x_{k-1}$. 
	\item Sample $a_1,\dots, a_{k-1} \sim \F_{q}\setminus \{0\}$ and set $x_k = \sum^{k-1}_{i}a_ix_i$
	
	\item If $f(x_k)=a_1f(x_1)+a_2f(x_2)+\dots + a_{k-1}f(x_{k-1})$ : return $1$; otherwise: return $0$
	
	%\item f $(f(x_1),\dots, f(x_k))\in \sfH_{\vec x}$
\end{itemize}

\end{tcolorbox}

\begin{remark}
Note that in the first step of the test, sampling $k-1$ random vectors instead of  $k-1$  linearly independent vectors changes the test passing probability by $O(q^{-n})$. This is because the total variation distance between two distributions is $O(q^{-n})$. Thus, it suffices to sample $k-1$ random vectors in the first step and carry out the remainder of the test as is. 
%This 
%For $k=3$, that is precisely the test proposed by Kiwi~\cite{Kiwi03} (for $q=2$, which is equivalent to the BLR test~\cite{BLR90}). 
\end{remark}

\paragraph{Comparison with Known Tests}
When $k=3$ and $q=2$, \hyperref[test:vspace]{$\test\_\mathsf{VSpace}_k(f)$} recovers the usual BLR test~\cite{BLR90}, and for arbitrary prime power $q$, it recovers the 3--query Kiwi test~\cite{Kiwi03}. The test that we obtain for $q=2$ and arbitrary $k$, has been studied by~\cite{KRV23, BKMR24,KMNR25} in a more restricted query model called the \textit{online manipulation-resilient testing model}, introduced by Kalemaj, Raskhodnikova and Varma~\cite{KRV23}.

% As stated before, we define $\delta_k(f)$ to be the ${\test\_\mathsf{VSpace}}_k$ passing probability of  $f$. Using this and~\cref{claim: Vexpression} we rewrite  $\sum_{\phi} \agrHom^k$. 
%\begin{equation}\label{eqn:Vexpression}
%	\sum_{\phi} \agrHom^k \approx_{q^{-n}} \parens[\Big]{1- \frac{q^k -1}{q^n(q-1)}    } \cdot q^{n-k} + q^{-(k-1)} \cdot    \sum^{k}_{j=1} \binom{k}{j} (q-1)^{j-1} \delta_j
%\end{equation}

\paragraph{Analysis of the Test} We state a simple combinatorial fact that we will need. 

\begin{fact}[Binomial Identity]\label{fact:binom}
Let $0 \leq j < k$ be any integers. Then,
	\[
 \sum^k_{i=j}(-1)^{k-i}  \binom{k}{i} \binom{i}{j} ~=~ 0\;.
	\]
\end{fact}
\begin{proof}
	A direct corollary of the fact that $\binom{k}{i} \binom{i}{j}  = \binom{k}{j} \binom{k-j}{i-j} $.
\end{proof}

\begin{corollary} 
	\label{cor:Vexptest}
	Let, $k\geq 1$ be any integer. 
  \[\sum_{\phi} \widetilde{\agr}(f,\phi)^k ~\approx_{q^{-n}}~ \frac{1}{q-1}(q\delta_k- 1). \] 
\end{corollary}

\begin{proof} Let, $\tilde{q}=q-1$. 
To compute $\sum_{\phi}  \widetilde{\agr}(f,\phi)^k$, we simply employ binomial expansion. 
%\begin{align*}
%	\sum_{\phi} (q\alpha_\phi)^i&~=~ \parens[\Big]{q^n - \frac{q^i -1}{q-1}  } + q \sum^i_{j=1} \binom{i}{j} (q-1)^{j-1} \delta_j && \text{By~\cref{eqn:Vexpression}}
	%	%\\&~=~  \parens[\Big]{q^n - \frac{q^i -1}{q-1}  }  + \frac{q}{q-1} \sum^i_{j=1} \binom{i}{j} (q-1)^{j-1} \delta_j 
%	\\&~=~ {q^n - \frac{q^i }{q-1}  }  + \frac{q}{q-1} \sum^i_{j=0} \binom{i}{j} (q-1)^{j} \delta_j 
%\end{align*}
%\[ \]
%Now we have,
%\fbox{Computing $\sum_\phi\tilde \alpha^k_\phi$ } 
\begin{align}
	(q-1)^k \cdot \sum_{\phi} \widetilde{\agr}(f,\phi)^k_{\phi}~&=~   \sum_{\phi } \parens{q\cdot \agrHom-1 }^k\\&~=~
		\sum^k_{i=0} \binom{k}{i} \cdot (-1)^{k-i}\cdot  \sum_{\phi} (q\cdot \agrHom)^i \notag 
		\\&~\approx_{q^{-n}}~ 	\sum^k_{i=0} \binom{k}{i} \cdot (-1)^{k-i}\cdot \parens[\bigg] {  q^n + \tq^{-1} - q^i \tq^{-1}  + q\tq^{-1} \sum^i_{j=0} \binom{i}{j} \tq^{j} \delta_j} .\label{eqn:intermediate}
		\end{align}
		The last line - \cref{eqn:intermediate} follows from~\cref{claim:tupleProb} and~\cref{claim: Vexpression}. It is easy to see that 
	\begin{align*}
		&\sum^k_{i=0} \binom{k}{i} (-1)^{k-i} (q^n+ \tq^{-1}) ~=~ 0, \qquad \text{and}\\
		&\sum^k_{i=0} \binom{k}{i}  (-1)^{k-i}q^i\tq^{-1} ~=~ -\tq^{k-1}.
	\end{align*}
%	$\sum^k_{i=0} \binom{k}{i} (-1)^{k-i} (q^n+ \tq^{-1}) =0$ $ \sum^k_{i=0} \binom{k}{i}  (-1)^{k-i}q^i\tq^{-1}= -\tq^{k-1}.$ 
	Replacing these to~\cref{eqn:intermediate} we get, 
\begin{align*}
	\text{\cref {eqn:intermediate}}~&=~  -\tq^{k-1} + q\tq^{-1}	   \sum^k_{i=0} 	\sum^i_{j=0}    (-1)^{k-i}  \binom{k}{i} \binom{i}{j} \tq^j\delta_j\; ,\\
	~&=~  -\tq^{k-1} + q\tq^{-1}	\sum^{k}_{j=0}  \tq^j\delta_j \sum^k_{i=j} 	    (-1)^{k-i}  \binom{k}{i} \binom{i}{j} \; ,
\\&~=~ -\tq^{k-1}+ q\tq^{k-1}\delta_k  && [\text{\cref{fact:binom}}]\; ,
\\&~=~\tq^{k-1}(q\delta_k- 1)\; .\qedhere
\end{align*}
\end{proof}

\begin{theorem}\label{theo:vspace}
 Let  $G=\F^n_{q}$ be a vector space and $H=\F_q$  for some finite field of order $q$. Let $k\geq 3$ be any odd integer. Then if $f$ passes \hyperref[test:vspace]{$\test\_\mathsf{VSpace}_k(f)$} with probability $\delta_k$, then, 
 \[ \frac{1}{q} + \parens[\big]{\frac{q-1}{q}} \parens[\Big]{\frac{q\delta_k-1}{q-1}}^{\frac{1}{k-2}} ~\leq~  \max_\phi \agr(f,\phi) \pm O(q^{-n})  ~\leq~  \frac{1}{q} + \parens[\big]{\frac{q-1}{q}} \parens[\Big]{\frac{q\delta_k-1}{q-1}}^{\frac{1}{k}} .\]
\end{theorem}
\begin{proof}
The upper bound is a direct consequence of~\cref{cor:Vexptest}. For the lower bound, let $k\geq 3$ be an odd integer,
	\[ \max_{\phi} \widetilde{\agr}(f,\alpha)^{k-2} \cdot    \sum_{\phi} \widetilde{\agr}(f,\phi)^2  ~\geq~     \sum_{\phi} \widetilde{\agr}(f,\alpha)^k ~\geq~   \frac{1}{q-1}(q\delta_k- 1)\; . \]
	Since $k-2$ is odd, this implies that we can take the positive $(k-2)^\mathrm{th}-$root on either side of the inequality.	
%	Here, the second inequality follows from~\cref{cor:Vexptest} and the test passing assumption. 
Thus, there exists a homomorphism $\phi$ such that 	
	\[
	\parens[\Big]{{\frac{q\cdot \agrHom- 1}{q-1}}}^{k-2} ~\geq~ \frac{1}{q-1}(q\delta_k- 1)\; .\qedhere   
	\]
%		and the theorem follows. 
\end{proof}

\subsection{Finite Field to Vector Space}
\label{subsec:cylic_vspace}
% {$\F_q$ to $\F_q^n$}
%\todo{add stuff}
Let,  $\F_{q}$ be a finite field of order $q$. 
Let, $G=\F_q$ and $H=\F^n_{q}$. Let $f: G \to H$ be an arbitrary function. The set of homomorphisms, $\Hom(G, H)$, have the property that for every non-trivial homomorphism $\phi$, $\ker(\phi) = \{0\}$. This property implies that no two homomorphisms agree on any non-zero input $x$. We note a simple consequence of this below.
%if $f(x) = \phi(x)$ for some $x, \phi$, 
%Let $k \geq 1$ and $\vec{x} \in \F_q^k\setminus \{0\}$ be a non-zero vector, let $x_i\neq 0$ be any non-zero element in the tuple. Then, for any $\phi \in \Hom(\F_q, \F_q^n)$, $\phi(x_i)\neq 0$
\begin{observation}\label{obs:triv_kernel}
	Let $k \geq 1$ and $\vec{x} \in \F_q^k\setminus \{\vec{0}\}$ be a non-zero vector. Then, $\ker(\Gamma_{\vec{x}}) = \{\triv\}$.
\end{observation}
\begin{proof}
Let $x_i\neq 0$ be any non-zero element in the tuple, $\vec{x}$. Then, for any non-trivial homomorphism $\phi \in \Hom(\F_q, \F_q^n)$, $\phi(x_i)\neq 0$ and thus $\phi \not \in \ker(\Gamma_{\vec{x}})$.
\end{proof}

Recall the general version of the test which samples $\vec{x}\in G^k \propto \abs{\ker{\Gamma_{\vec{x}}}}$. This is problematic in our case as for $\vec{x} = 0$, we have $\abs{\ker{\Gamma_{\vec{0}}}} = q^n$, but $\ker(\Gamma_{\vec{x}}) = 1$ for every other vector $\vec{x}$. We circumvent this issue by simply ignoring the $0$ element in $G$ and working over $\widetilde{G} = \F_q^*$, the set of non-zero elements of $\F_q$. Accordingly, we
will work with the fractional agreement of $f$ over $\tilde{G}$,  \[ \tagrHom := \Ex{x\sim \widetilde{G}}{\indicator{f(x) = \phi(x)}}.\]
%\paragraph{Analyzing the expression}
%Let We will work with the fractional agreement of $f$ over $\tilde{G}$,  $\tagrHom := \Ex{x\sim \widetilde{G}}{\indicator{f(x) = \phi(x)}}$.
Using this modified agreement, $ \tagrHom $, we have the following claim. 
\begin{lemma}[{{A variant of~\cref{lem:alpha_k}}}]\label{lem:fqk} Let $f: G\to H$ to be any function. Then, for any $k\geq 1$, 
	\[
	\sum_{\phi\in \Hom(\F_q, \F_q^n)} \tagrHom^k ~=~  \Ex{\vec{x} \sim \widetilde{G}^k }{ \indicator{f(\vec x) \in \sfH_{\vec x}} } .
	\]	
\end{lemma}
%\snote{I think this proof is kinda redundant. Maybe write the general case (up) for set $S$ and then take $S$ to be $G$}
\begin{proof}
	Identical to the proof of~\cref{lem:alpha_k}, after using~\cref{obs:triv_kernel}. 
\end{proof}
The expression in~\cref{lem:fqk} suggests the following simple test. 
%Recall the %general version of the test which samples $\vec{x}\in \F_q^k \sim \abs{\ker{\Gamma_{\vec{x}}}}$. This is problematic in our case as for $\vec{x} = 0$, we have $\abs{\ker{\Gamma_{\vec{0}}}} = q^n$, but $\ker(\Gamma_{\vec{x}}) = 1$ for every other vector $\vec{x}$. 
%This expression, however, suggests that we need to exclude the all-zeros vector, and we define the following variant that only samples a tuple from $\widetilde{G} = \F_q^*$, the set of non-zero elements of $\F_q$.
 
\vspace{1em}
	\begin{tcolorbox}[colframe=teal, colback=white, title=${\sf{Test\_NonZero}}_k (f)$, label=test:fqVspace]
	\begin{itemize}
		\item Sample $\vec x\sim \widetilde{G}^k$ uniformly.
		\item If $f(\vec x)\in \sfH_{\vec x}$: return $1$; otherwise: return $0$. 
		\item Equivalently, the test passes only if $x_i^{-1}f(x_i) = x_j^{-1}f(x_j)$ for every $x_i, x_j$ in the sampled tuple $\vec{x}$.
	\end{itemize}
\end{tcolorbox}

\begin{theorem}\label{theo:fq_vspace}
Let  $G=\F_{q}$ some finite field of order $q$  and $H=\F_q^n$ be a vector space.  Let, $k\geq 2$ be any integer. Then if $f: G\to H$ passes
\hyperref[test:fqVspace]{${\sf{Test\_NonZero}}_k (f)$}
  with probability $\delta_k$, then,
 \[\parens[\Big]{1-\frac{ 1}{q}}\cdot  \delta_k^{\frac{1}{k-1}} ~\leq~ \max_\phi \agr(f,\phi) ~\leq~ \frac{1}{q} + \parens[\Big]{\frac{q-1}{q}}\cdot \delta_k^{\frac{1}{k}}   \;\;.\]
\end{theorem}
\begin{proof}
From the definition of the shifted agreement, we have, for any $\phi$\
 \[\agrHom ~\leq~  \frac{1}{q} + \frac{q-1}{q}\cdot \tagrHom. \] 
Since, $ \max_{\phi} \tagrHom^k ~\leq~ \sum_\phi \tagrHom^k ~=~ \delta_k$, we get the upper bound. We have, \[ \max_{\phi} \{\widetilde{\agr}(f,\phi)^{k-1}\}\cdot \sum_{\phi} \widetilde{\agr}(f,\phi) ~\geq~  \sum_{\phi} \widetilde{\agr}^k(f,\phi) ~=~ \delta_k \; . \]
From~\cref{lem:fqk} for $k=1$, we get $\sum_{\phi} \widetilde{\agr}_{\phi} \leq 1$. Clearly, this quantity is greater than $0$, as otherwise the above inequality forces $\delta_k = 0$ for which the theorem trivially holds. Thus, we have the inequality,
\[  \tagrHom ~\geq~  \delta_k^{\frac{1}{k-1}}. \]
Finally, by rescaling we get or result:
\[\agrHom ~\geq~  \frac{|\tilde{G}|}{|G|}\cdot  \tagrHom ~=~ \parens[\Big]{1-\frac{ 1}{q}}\cdot \delta_k^{\frac{1}{k-1}}. \qedhere \]
\end{proof}

\section{Automorphism testing for Non-Abelian groups}\label{sec:auto}
\subsection{Automorphism testing over Dihedral groups}\label{subsec:dihedral}
We begin by deriving an identity, analogous to~\cref{lem:alpha_k}, for $\sum_{\phi\in \Aut(G)}  \agrHom^k$ that holds for any finite group $G$. Then we use that identity to construct a test for the dihedral group, $\D_{2p}$ of order $2p$.
%Let $G$ be a finite group (not necessarily abelian) and $f: G\rightarrow G$ be a function. We want to design a test to estimate $f$'s distance from $\Aut(G)$. 
\paragraph{Analyzing the key expression. }
Recall that, in the case of homomorphisms between abelian groups, $(G,H)$, to obtain such an identity, we used the fact that,  evaluation map: $\Gamma_{\vec x}: \Hom(G,H)\rightarrow H^k$ is a homomorphism (therefore, an $N$-to-one map). 

The evaluation map $\Gamma_{\vec x}$ is well defined for any set of functions and consequently, can be defined for $\Aut(G)$ in an expected manner: 
	\[ \Gamma_{\vec x}: \Aut(G)\rightarrow G^k~:~\Gamma_{\vec x}(\phi)=\parens{\phi(x_1),\dots,\phi(x_k)}\;. \]
The set $\Aut(G)$ is typically a non-abelian group under composition, and $G$ is an arbitrary (not necessarily abelian) finite group. As a result,  $\Gamma_{\vec x}$ is typically not a homomorphism anymore. %When $G$ is an arbitrary (not necessarily) finite group and we are working with $\Aut(G)$, the evaluation map $\Gamma_{\vec x}: \Aut(G)\to G^k$ is not a homorphism anymore. 
However, the next lemma shows that it is still 
	a $N$-to-one map for $N=\stab(\vec x)$ where
	$\stab(\vec x)$ 
	 is the pointwise stabilizer subgroup of the set $\{x_1,\dots, x_n \}\subseteq G$, i.e., 
	 \begin{equation}
	 	\stab(\vec x):=\{\phi \in \Aut(G)~:~\phi(x_i)=x_i~\text{ for }i=1,\dots,n \} \label{eqn:stab} \;.
	 \end{equation}
	 %\[ .\] 
We define $\sfG_{\vec x} := \Ima\parens{\Gamma_{\vec x}}$.

\begin{lemma}\label{lem:aut_exp}
	Let, $G$ be any finite group and $\Aut(G)$ be the group of automorphims of $G$. Let $f: G\to G$ be a function.  Let, $k\geq 1$ be any integer. Then,
	\[
	\sum_{\phi\in \Aut(G)}  \agrHom^k ~=~ \Ex{\vec{x} \sim G^k}{ \indicator{f(\vec x) \in {\sf G}_{\vec x}} \abs{\stab(\vec x)}} \;.
	\]		
\end{lemma}
\begin{proof}
Following the initial steps as in~\cref{lem:alpha_k}, we have: 
	\begin{align*}
		\sum_{\phi\in \Aut(G)}	 \agrHom^k ~&=~  \Ex{\vec{x} \sim G^k}{ \sum_{\phi\in \Aut(G)}	  \indicator{f(\vec x) = \phi(\vec x)} }.
	\end{align*}
 Note that any tuple $\vec x\in G^k$ satisfies:
	\[
	\sum_{\phi\in \Aut(G)}	  \indicator{f(\vec x) = \phi(\vec x)} = 
	\begin{cases}
		\abs{\{ \phi~:~\phi(\vec x)=f(\vec x) \}}  & \text{if } f(\vec x) \in {\sf G}_{\vec x}, \\
	 0 & \text{if } f(\vec x) \notin {\sf G}_{\vec x} .
	\end{cases}
	\]
	Thus, it suffices to focus on $\vec x\in {\sf G}_{\vec x} $ case. Fix such $\vec x$ and define $\Phi_{\vec x}$ to be the set of automorphisms that evalutes to $f(\vec x)$ for input $\vec x$, i.e., 
	\[ \Phi_{\vec x} = \{ \phi~:~\phi(\vec x)=f(\vec x)\}.\]  If two distinct automorphisms $\psi,~\theta\in \Phi_{\vec x}$, then we have:
	\[\psi(\vec x)=\theta(\vec x) \iff \theta^{-1}\psi(\vec x)=\vec x  \iff  \theta^{-1}\psi \in \stab(\vec x), \]
	where $\stab(\vec x)=\{\phi \in \Aut(G)~:~\phi(x_i)=x_i~\text{ for }i=1,\dots,n \}$ is the pointwise stabilizer subgroup of the set $\{x_1,\dots, x_n \}\subseteq G$. 
	As, $f(\vec x)\in {\sf G}_{\vec x} $, there is at least one 
	automorphism, $\psi$ in $\Phi_{\vec x}$, and we can write: 
	\[ \Phi_{\vec x} ~=~ \{ \psi\sigma~:~\sigma\in \stab(\vec x) \}\;\;.\]
	For any two distinct automorphisms, $\sigma$, $\sigma'$ it holds that 
	$\psi\sigma\neq \psi\sigma'$ are distinct. This implies $\Phi_{\vec x}=\abs{\stab(\vec x)}$. Thus, for any $\vec x$ such that $f(\vec x)\in{\sf G}_{\vec x} $, it holds that 
	\[	\sum_{\phi\in \Aut(G)}	  \indicator{f(\vec x) = \phi(\vec x)} = |\Phi_{\vec x}|=|\stab(\vec x)|. \qedhere\] 
\end{proof}
\begin{remark}
	In~\cref{lem:aut_exp}, we did not partition the final expression based on the condition $\cG_k=G^k$, as was done in~\cref{lem:alpha_k}. Although a similar partitioning could be applied over as well, the simpler form is sufficient for our application to the dihedral group. 
	
\end{remark}

\paragraph{Defining the test for Dihedral group.}
Now we focus on the dihedral case. Let, $G=\D_{2p}$ to be dihedral group of order $2p$ for some prime $p$. Motivated by the expression in~\cref{lem:aut_exp}, we define the following test which is analogous to~\hyperref[test:generic]{$\test\_\ker_k$}  :
%define the distribution $\cD_k$ on $G^k$ which samples  $\vec x \propto |\stab({\vec x})|$. For this, we define the quantity $\rho_k$ analgous to $\gamma_k$ in~\cref{clm:gamma1},
%We have the following test.  %\ie 
%\[
%\cD_{\stab}(\vec x) ~:=~ 
%\frac{ |\stab(\vec x)| }{\sum_{\vec x \in G^k} | \stab(\vec x)|}.\]

\vspace{1em}
	\begin{tcolorbox}[colframe=teal, colback=white, title={$\test\_ {\sf Dihedral}_k(f)$}, label=test:dihedral]
	\begin{itemize}
		\item Sample $\vec x\propto |\stab(\vec x)| $.
		\item If $f(\vec x)\in \mathrm{Im}(\Gamma_{\vec x}) =  \sfG_{\vec x}$: return $1$; otherwise: return $0$.
	\end{itemize}
\end{tcolorbox}
\vspace{1em}

\begin{claim}\label{claim:dihedral}
	If $f: G\rightarrow G$ passes $\test\_ \Aut_k (f)$ with probability $\delta_k(f)$, then it holds
	\[\sum_{\phi} \agrHom^k ~=~ \delta_k(f) \cdot \frac{\sum_{\vec x \in G^k} |\stab(\vec x)|  }{|G|^k} ~=~ \frac{\rho_k\delta_k(f)}{|G|^k}  \;\;.\] where $\rho_k:= \sum_{\vec x \in G^k} |\stab(\vec x)|$.
\end{claim}
\begin{proof}
	Follows directly from the test definition and~\cref{lem:aut_exp}.
\end{proof}

\paragraph{Computing the quantity $\rho_k$.}For any automorphism $\phi\in \Aut(G)$, let $\fix(\phi)$  be the set of fixed points of $\phi$, i.e.,  $\fix(\phi)=\{x\mid\phi(x) = x \}$. We have the following claim, which provides an alternative expression for $\rho_k$ in terms of the fixed points of automorphisms. This claim is similar to~\cref{clm:gamma1}. 
\begin{claim}
	\label{clm:fixedpoints} 
For any finite group $G$, and integer $k\geq 1$, we have	
\[
\rho_k := \sum_{\vec x \in G^k}{\abs*{\stab(\Gamma_{\vec x})}} ~~=~ \sum_{\phi \in \Aut(G)}{\abs{\fix(\phi)}^k}  . \]
\end{claim}

\begin{proof}
Recall, $\stab(\vec x)$ 
is the pointwise stabilizer subgroup of the set $\{x_1,\dots, x_n \}\subseteq G$. So from the definition, we get: 
\begin{align*}
	\sum_{ {\vec x} \in G^k }  |\stab(\vec x)|
	&~=~\sum_{ {\vec x} \in G^k } \sum_{\phi\in \Aut(G)} \indicator{\phi(x_i)=x_i~\forall i } &&[\text{By~\cref{eqn:stab}}]
	\\&~=~ \sum_{\phi} \sum_{\vec x \in G^k } \indicator{\phi(x_i)=x_i~\forall i }&&\text{[By Fubini]}
	\\&~=~ \sum_\phi \abs{\fix(\phi)}^k . \qedhere
\end{align*}

\end{proof}
Now we compute bounds on the quantity $\rho_k$ for the dihedral group, $\D_{2p}$. To do such, we will need the following few basic facts about dihedral groups.

\begin{fact}[Dihedral group and its automorphisms]	\label{fact:aut_D}
	Let $n\geq 3$ be any integer. The dihedral group of order $2n$, $\D_{2n}$, and its automorphism group are defined as follows: 
	\begin{align}
		D_{2n} ~&=~ \langle r,s \mid s^2=e,~ r^n=e,~srs=s^{-1}\; \rangle,  \nonumber \\
		\Aut(\D_{2n}) ~&=~ \braces[\big]{\phi_{\ell,m}~:~0\leq m\leq n-1,~1\leq \ell\leq n-1 \text{ and } \gcd(n,k)=1  },\nonumber
			\\[1.5pt] &\quad \text{ where,   }\phi_{\ell,m}(r)~=~r^\ell,~\phi_{\ell,m}(s)~=~sr^m .\label{eqn:aut_action}
	\end{align}
\end{fact}
%We also have the following explicit description of $(\D_{2n})$. 
%
%\begin{fact}[Description of automorphism group]
%	\label{fact:aut_D}
%		Let, $n\geq 3$ be any integer and $\D_{2n}$ be the automorphism group of order $2n$. Then, 
%		\begin{align}
%			&\Aut(\D_{2n})= \{ \phi_{k,m}~:~0\leq m\leq n-1,~1\leq k\leq n-1 \text{ and } \gcd(n,k)=1,   \}
%			\\ &\text{ where,   }\phi_{k,m}(r)~=~r^k,~\phi_{k,m}(s)~=~sr^m .\label{eqn:aut_action}
%		\end{align}
%		%\begin{equation}
%		%	\
%	%	\end{equation}
%		%\[  \]
%
%%\begin{equation}\label{eqn:aut_action}
%%\end{equation}
%%$~\text{ where, }$ 
%\end{fact}
Using this presentation, the group $\D_{2n}$ can be written as $\D_{2n}=\text{ Rotations}\;\cup \;\text{Reflections}$, where: 
\[	\text{Rotations}=\{e,r,\dots, r^{n-1} \},~	\text{Reflections}=\{s,sr,\dots, sr^{n-1} \}.
\]

Now we have all the required tools to bound the quantity $\rho_k$. 
\begin{claim}
\label{clm:rhok}
	Let, $p>3$ be any prime. Let $G=\D_{2p}$ be the dihedral group of order $2p$. Let, $k\geq 2$ be any integer. Then,
	\[p^k\parens[\big]{(p-1)+2^{k}}~\leq~  \rho_k~\leq~ p^k\parens[\big]{(p-1)+2^{k+1}}\;\;. \]
\end{claim}

\begin{proof} From~\cref{fact:aut_D}, we know that any element of $\Aut(\D_{2p})$ is of the form $\phi_{\ell,m}$ 
for some $m\leq n-1$ and $1\leq \ell\leq n-1$. Consider any such $\phi_{\ell,m}$. If it fixes a rotation element, $r^i$, then it must satisfy: $\phi_{\ell,m}(r^i)=r^i$. 
Similarly, if it fixes a reflection element, $sr^j$, then it must hold that: $\phi_{\ell,m}(sr^j)=sr^j$. These imply the following linear congruence relations.  
%\begin{equation}
%	\label{eqn:cong1} 
% \phi_{k,m}(r^i)=r^{i} \iff r^{ik}= r^i \iff i(k-1) \equiv 0 \pmod p 
%\end{equation}
\begin{align}
	 \phi_{\ell,m}(r^i)= &r^{i} \iff r^{i\ell}= r^i~~ ~~\text{[By \cref{eqn:aut_action}]} \iff i(k-1) \equiv 0 \pmod p  \label{eqn:cong1} .\\
		\phi_{\ell,m}(sr^j)= &sr^j \iff sr^{m}r^{j\ell} = sr^j  ~~\text{[By \cref{eqn:aut_action}]} \iff m \equiv j(1-\ell)\pmod p  .\label{eqn:cong2}
\end{align}
%\begin{equation}
	%	\label{eqn:cong2} 
%	\phi_{k,m}(sr^j)=sr^j \iff sr^{m}r^{jk} = sr^j \iff m \equiv j(1-k)\pmod p 
%\end{equation}
Thus for any $\phi_{\ell,m}$, the number of fixed point is the following quantity:
\[ \fix(\phi_{\ell,m})= \text{\# of solutions to~\cref{eqn:cong1}} +  \text{\# of solutions to~\cref{eqn:cong2}} \;.\]
	
\begin{itemize}
	\item Case 1: $\ell\neq 1,~m=0$.  Since $\ell-1\neq 0\pmod p$, it is invertible. Therefore,
	\[
		   \quad i(\ell-1) ~\equiv~ 0 \implies i \equiv 0, \quad \text{and similarly, }\;
		    -j(\ell-1) ~\equiv~ m ~\equiv~ 0 \implies j
		   \equiv 0.
		  \] 
	 This gives $\fix(\phi_{\ell,m})=2$.  
	
	\item Case 2: $\ell\neq 1,~m\neq 0$. 
	As in the first case,~\cref{eqn:cong1} has one solution. ~\cref{eqn:cong2} also has one solution that is $j\equiv m(1-\ell)^{-1} \pmod p$. So here also
	   we have: $\fix(\phi_{\ell,m})=2$.  
	
		\item Case 3:  $\ell=1,~m\neq 0. $ If $\ell=1$, then~\cref{eqn:cong1} is always satisified regardless of value of $i$. Therefore, there are $p$ solutions as every rotation gets fixed. On the other hand, if $\ell=1$ then the RHS of~\cref{eqn:cong2} is zero, whereas the left-hand side is $m\neq 0$. So, there is no solution to~\cref{eqn:cong2}. It follows that $\fix(\phi_{\ell,m})=p$.  
		
		\item Case 4:  $\ell=1,~m=0.$
		 By a similar argument as in case 3, we get that: $\fix(\phi_{\ell,m})=2p$.  

\end{itemize}
Combining the counts from the cases above, we get,
\begin{align}
	\rho_k&~=~~ \sum_\phi \abs{\fix(\phi)}^k \nonumber
	\\&~=~\sum_{\substack{2\leq \ell\leq p-1\\ 1\leq m\leq p-1}}  |\fix(\phi_{\ell,m})|^k + 
	\sum_{\substack{2\leq \ell\leq p-1\\ m=0}}  |\fix(\phi_{\ell,m})|^k + 
	\sum_{\substack{\ell=1\\ 1\leq m\leq p-1}}  |\fix(\phi_{\ell,m})|^k + |\fix(\phi_{1,0})|^k
	\nonumber
	%~\\&~=~ \sum_{\substack{2\leq k\leq p-1\\ 1\leq m\leq p-1}} 2^k
	%+ (p-2)2^k + (p-1)p^k+ (2p)^k 
	~\\&~=~ (p-1)(p-2)2^k+(p-2)2^k + (p-1)p^k+ (2p)^k  
	\nonumber
	~\\&~=~(p-1)p^k+2^k\parens[\Big]{p(p-2)+p^k}\label{eqn:rhok}
	~\\&~\leq~(p-1)p^k+2^k\cdot(2p^k)
	\nonumber
	~\\&~= p^k\parens[\big]{(p-1)+2^{k+1}}
	\nonumber
\end{align}
For the inequality, we have used the assumption that $k\geq 2$. For the lower bound we simply take the terms involving $p^k$ in the expression given by~\cref{eqn:rhok}, giving us:
\[\rho_k ~=~ (p-1)(p-2)2^k+(p-2)2^k + (p-1)p^k+ (2p)^k~\geq~ p^k\parens[\big]{(p-1)+2^k}.\qedhere\]
\end{proof}

We can now use the above calculation for $\rho_k$ to deduce a testing result. 

% and a (combinatorial) list decoding bound. 

\begin{theorem}[Testing $\Aut(D_{2p})$ ]\label{thm:dihedral_aut}
	Let  $G=\D_{2p}$ be the dihedral group of order $2p$ for some prime $p>3$. Let $k \geq 3$ be an integer, and $f: G\to G$ be any function. Then if $f$ passes \hyperref[test:dihedral]{$\test\_ {\sf Dihedral}_k$ }
	with probability $\delta_k(f)$, then there exists a automorphism $\phi \in \Aut(G)$ such that $\agr(f,\phi) ~\geq~  \frac{1}{2} \cdot \delta_k(f)^{\frac{1}{k-2}}  $ .
\end{theorem}
\begin{proof}Using~\cref{eq:max_approx} and~\cref{claim:dihedral}, we have: 	\[	\max_{\phi}\, \agr(f,\phi) ~\geq~  \frac{\sum_{\phi}{\agr(f,\phi)^i}}{\sum_{\phi}{\agr(f,\phi)^{i-1}}} ~\geq~  \frac{\delta_i \rho_i}{|G| \delta_{i-1} \rho_{i-1}} .	
	\]
Multiplying this for $i \in [3, k]$, we get
\begin{align*}
	\parens[\Big]{\max_{\phi} \agr(f,\phi)}^{k-2} ~&\geq~ \parens[\Big]{\frac{1}{|G|}}^{k-2} \cdot \frac{\delta_k \rho_k}{\delta_{2} \rho_{2}}\\
	~&\geq~ \parens[\Big]{\frac{1}{|G|}}^{k-2} \cdot \frac{\delta_k \rho_k}{\rho_{2}}  && [\text{Since, }  \delta_{2} \leq 1.]\\
	~&\geq~  \delta_k  \cdot \frac{1}{2^{k-2}} \frac{(p-1)+2^k}{(p-1)+2^3}  && [\text{Using~\cref{clm:rhok}}]\\
	~&\geq~  \delta_k  \cdot \frac{1}{ 2^{k-2} } &&[k\geq 3] . \qedhere
\end{align*} 
\end{proof}
\begin{remark}
As a contrast to our other results, \cref{thm:dihedral_aut} does not give a group-independent upper bound on the maximum agreement, but instead gives a very weak $O((p\delta)^{\frac{1}{k}}$- bound. This is because soundness guarantee only requires a bound on $\frac{\rho_i}{|G|\rho_{i-1}}$, but to get an upper bound, one needs a bound on ${\rho_i}$ which is not true here.  	
\end{remark}

\subsection{Inner Automorphism Testing}
While it is hard to know the structure of $\Aut(G)$ for a general $G$, there is a canonical subgroup of automorphisms that can be easily described. Let $G$ be any group and let $\Inn(G) \subseteq \Aut(G)$ be the subset of \textit{inner automorphisms}, \ie  $\braces{\phi_g \mid  G\to G}$. In this section, we will show that our framework yields tests for $\Inn(G)$ for many families of groups.  

\paragraph{Defining the test.} The setup of Automorphism testing, as in the previous section, works almost identically to test inner automorphisms. The only thing that changes is the computation of (an analog of) $\rho_k$. Naturally , we need to analyze the map,

\[ \Gamma_{\vec x}: \Inn(G)\rightarrow G^k~:~\Gamma_{\vec x}(\phi)=\parens{\phi(x_1),\dots,\phi(x_k)} \;.  \]
Again this is a $N$-to-one map, where $N = |\Inn\stab(\vec x)|$, and $\Inn\stab(\vec{x})$ is defined as:
\[
\Inn\stab((x_1,\dots, x_k)) ~:=~ \braces[\big]{\phi_g \in \Inn(G) \mid \phi_g(x_i) = x_i \;\forall i \in [k]} ~=~ \bigcap_{i \in [k]}  C_G(x_i) \;.
\]
This can be seen by observing that if $\phi_g(x_i) = \phi_h(x_i)$ for all $i$, if and only if $\phi_{h^{-1}g} \in \Inn\stab(\vec{x})$. Since, $\Inn\stab$ is a group,~\cref{lem:aut_exp} generalizes directly to:
 	\begin{equation}\label{eq:agreement_inner}
 		\sum_{\phi\in \Inn(G)}  \agrHom^k ~=~ \Ex{\vec{x} \sim G^k}{ \indicator{f(\vec x) \in \Ima(\Gamma_{\vec x})}\cdot  \abs{\Inn\stab(\vec x)}} \;\;. 
	 	\end{equation}	
This yields the following test which is analogous to~\hyperref[test:dihedral]{$\test\_ {\sf Dihedral}_k$}  :
%define the distribution $\cD_k$ on $G^k$ which samples  $\vec x \propto |\stab({\vec x})|$. For this, we define the quantity $\rho_k$ analgous to $\gamma_k$ in~\cref{clm:gamma1},
%We have the following test.  %\ie 
%\[
%\cD_{\stab}(\vec x) ~:=~ 
%\frac{ |\stab(\vec x)| }{\sum_{\vec x \in G^k} | \stab(\vec x)|}.\]

	\begin{tcolorbox}[colframe=teal, colback=white, title={$\test\_ {\sf Inner}_k(f)$}, label=test:inner]
	\begin{itemize}
		\item Sample $\vec x\propto |\Inn\stab(\vec x)| $.
		\item If $f(\vec x)\in  \mathrm{Im}(\Gamma_{\vec x})$\,: return $1$; otherwise: return $0$.
	\end{itemize}
\end{tcolorbox}
%\vspace{0.4em}
\begin{remark}
The test requires one to check if there exists a $g\in G$ such that $f(x_i) = gx_i g^{-1}$ for every $i$. This is known as the \textit{simultaneous conjugacy problem}, and for groups such as the symmetric group and matrix groups, it can be solved efficiently.  We do not delve into these details as we are only concerned with the query complexity of the test.
\end{remark}

\begin{claim}\label{claim:inner_aut}
	If $f: G\rightarrow G$ passes $\test\_ \sf Inner_k (f)$ with probability $\delta_k(f)$, then it holds
	\[\sum_{\phi \in \Inn(G)} \agrHom^k ~=~ \delta_k(f) \cdot \frac{\sum_{\vec x \in G^k} |\Inn\stab(\vec x)|  }{|G|^k} ~=~ \frac{\trho_k\delta_k(f)}{|G|^k}  \;\;.\] where $\trho_k:= \sum_{\vec x \in G^k} |\Inn\stab(\vec x)|$.
\end{claim}
\begin{proof}
	Follows directly from the test definition and~\cref{eq:agreement_inner}.
\end{proof}

\begin{lemma}\label{lem:geninn_test}
	Let $G$ be a group, and $\tau \geq 2$ be an integer, such that $\frac{\trho_i}{|G|\trho_{i-1}} \geq c$ for every $i \geq \tau$. Let $k \geq \tau$ be an integer, and $f: G\to G$ be any function. Then if $f$ passes \hyperref[test:inner]{$\test\_ {\sf Inner}_k$ }
	with probability $\delta_k(f)$, then there exists a automorphism $\phi \in \Inn(G)$ such that $\agr(f,\phi) ~\geq~ c \cdot \delta_k(f)^{\frac{1}{k-\tau+1}}  $ .
\end{lemma} 
\begin{proof}
	Identical to the proof of~\cref{thm:dihedral_aut}.
\end{proof}

\paragraph{Computing the quantity $\trho_k$.}
The quantity $\trho_k$ is directly related to the sizes of the \textit{centralizer} and the \textit{conjugacy classes} of the group. To compute this, we first define the relevant entities. We will then compute $\trho_k$ using direct computation (for the symmetric group and quasirandom groups), and by using known results about these quantities from existing results on finite simple groups. 

\begin{fact}[Centralizer and Center]\label{fact:centralizer}
	For any  $g \in G$, $C_G(g) = \fix(\phi_g) = \{ x \mid gx = xg \}$. Moreover, if $\mathcal{C}_g$ is the conjugacy class of $g$, then, $\phi_g: G/C_G(g) \to \mathcal{C}_g $ is a bijection and thus,  
	\[\abs{C_G(g)} ~=~ \abs{\fix(\phi_g) } ~=~ \frac{|G|}{|\mathcal{C}_g|}.\]
Let $Z(G) = \{g \mid gx = xg \;\forall x \in G\}$, be the center. Thus, $\phi_g$ is an identity homomorphism if and only if $g \in Z(G)$. Therefore, $\Inn(G) \cong G/Z(G)$.  
\end{fact}

We now derive the main expressions we will use to compute $\trho_k$. This simple but crucial lemma connects our analysis of the test with group-theoretic \enquote{zeta functions}.

\begin{corollary}\label{cor:trho_expression}
	For any group $G$, let $\mathcal{C}(G)$ denote its conjugacy classes. Then, 
	\[
	\trho_k(G) ~:=~ \sum_{\vec x \in G^k} |\Inn\stab(\vec x)| ~=~ \sum_{\phi_g \in \Inn(G)} \abs{\fix(\phi_g)}^k ~=~ \frac{|G|^k}{|Z(G)|} \cdot   \sum_{\mathcal{C} \in \mathcal{C}(G)} \abs{\mathcal{C}}^{1-k}\; .
	\]
\end{corollary}
\begin{proof} The proof of the first equality is identical to that of~\cref{clm:fixedpoints}, and so we omit it.  Now, $\phi_g = \phi_{ga}$ for any $a \in Z(G)$, and thus, replacing the summation by $g \in G$ modifies it by a factor of $|Z(G)|$. Using this,
	\begin{align*}
		|Z(G)| \cdot \rho_k(G) ~&=~ \sum_{g\in G} \abs{\fix(\phi_g)}^k\\
		~&=~ |G|^k \cdot \sum_g \abs{\mathcal{C}_g}^{-k}  && [\text{\cref{fact:centralizer}}]\\
		~&=~ |G|^k \cdot \sum_{\mathcal{C} \in \mathcal{C}(G)} \abs{\mathcal{C}}^{1-k} \; . \qedhere 
	\end{align*}
\end{proof}

The quantity $\eta^G(k-1) := \sum_{\mathcal{C} \in \mathcal{C}(G)} \abs{\mathcal{C}}^{1-k}$  has been studied by~\cite{LS05}, and we note a corollary of their general result:
%\begin{theorem}\textup{\cite[Thm 1.10]{LS05}}
%	Let $L(q)$ be the finite simple group of type $L$ over $\F_q$. Then,
%\begin{itemize}
%	\item (Bounded Rank) For any $t >1$, $\eta^{L(q)} (t) \to 1$ as $q \to \infty$.
%	\item (Unbounded Rank) Fix any $t >1$. Then there exists $r(t)$ such that for every simple group $L(q)$ of rank $r \geq r(t)$, $\eta^{L(q)} (t) \to 1$ as $|L(q)| \to \infty$.
%\end{itemize}	
%\end{theorem}

\begin{theorem}\textup{\cite[Cor. 5.1]{LS05}}\label{theo:liebeck}
	Let $t > \frac{1}{4}$. For every finite simple group, $G$, except $\mathrm{PSL}_2(q)$, $\mathrm{PSL}_3(q),$ $\mathrm{PSU}_3(q)$, we have, $\eta^{G} (t) \leq 1 + o_{|G|}(1)$.
\end{theorem}

Their result also works for a subclass of \textit{almost simple groups}. The above result gives us $k$-query tests for these families of groups, for any $k \geq 3$. Instead of solely using this general result, which uses deep results from Deligne-Lustzig theory, we will give elementary proofs for quasirandom groups and symmetric/alternating groups that will cover the bounded rank case (at the cost of a larger, but constant, query complexity). 
%Our approach will also give fairly explicit bounds instead of an asymptotic guarantee.

%wherein a group $H$ is called \textit{almost simple} if there exists a simple group $G$ such that $G\subseteq H\subseteq \Aut(G)$.

\subsubsection{Symmetric (and Alternating) Group}
The conjugacy classes of the symmetric group correspond to cycle types where the cycle type of a permutation $g \in \Sym_n$, is given by the number of cycles of length $i$, which we denote as $a_i(g)$. The following is a known fact,
\[
|\mathcal{C}_g| = \frac{n!}{\prod_i i^{a_i(g)}\cdot a_i(g)!}, \quad \text{and thus,}\quad \abs{\fix(\phi_g)} = |C_{S_n}(g)| = \prod_i i^{a_i(g)}\cdot a_i(g)! \;\; .
\]

We now carry out the computation of $\trho_k(\Sym_n)$ for any $k \geq 2$. The bound for $k=2$ is stated in the sequence A110143 in OEIS (or \cite{stackex}) without any proof.

\begin{lemma}\label{lem:rho_sym}
	For $G = \Sym_n$, and $k\geq 2$,
	\[
	  1 ~\leq~ \frac{\trho_k(G)}{|G|^k}  ~\leq~ 1 + O\parens[\Big]{n^{-2(k-1)}} \;.
	\]
\end{lemma}
\begin{proof}
The formula for the sizes of the conjugacy classes shows that $Z(G) = \{1\}$ as there is only one conjugacy class of size $1$. Now using~\cref{cor:trho_expression}, and looking at the trivial conjugacy class, one gets the lower bound. We will now compute the upper bound: 
	\begin{align*}
	\trho_k ~&=~ \sum_{g} \parens[\Big]{\prod_i i^{a_i(g)}\cdot a_i(g)!}^k\\
	~&=~ \sum_{t=0}^{n} \sum_{g: a_1(g) = t}   \parens[\Big]{t!\cdot \,\prod_{i \geq 2} i^{a_i(g)}\cdot a_i(g)!\,}^k\\
	~&=~ \sum_{t=0}^{n} (t!)^k \sum_{g: a_1(g) = t}   \parens[\Big]{ \,\prod_{i \geq 2} i^{a_i(g)}\cdot a_i(g)!\,}^k\\
	~&\leq~ (n!)^k + \sum_{t=0}^{n-2} (t!)^k \sum_{g: a_1(g) = t}   \parens[\Big]{ \,\prod_{i \geq 2} i^{a_i(g)}\cdot a_i(g)^{a_i(g)}\,}^k &&[x!\leq x^x]\\
	~&=~ (n!)^k + \sum_{t=0}^{n-2} (t!)^k \sum_{g: a_1(g) = t}   {  e^{k\cdot \sum_{i \geq 2}a_i(g)\log(i\cdot a_i(g))}\,} \; .
%		~&=~ \sum_{t=0}^{n} \sum_{g, a_1(g) = t} {n \choose t}  \parens[\Big]{\,\prod_{i \geq 2} i^{a_i(h)}\cdot a_i(h)!\,}^k\\
%	~&=~ (n!)^k+ \sum_{t=0}^{n-2} {n \choose t} \sum_{h \in S_{n-t},} \parens[\Big]{\prod_{i\geq 2} i^{a_i(h)}\cdot a_i(h)!}^k\\
%	~&\leq~ (n!)^k+ \sum_{t=0}^{n-2} {n \choose t} \sum_{h \in S_{n-t}} \parens[\Big]{\prod_{i>2} i^{a_i(h)}\cdot a_i(h)^{a_i(h)}}^k\\
%	~&=~  (n!)^k+ \sum_{t=0}^{n-2} {n \choose t} \sum_{h \in S_{n-t}} \parens[\Big]{ e^{\sum_{i \geq 2} a_i\log (a_i(h)\cdot i)}}^k\\
\end{align*} 

Now, for any $g\in S_{n}$, we have $\sum_{i=1}^n, i \cdot a_i(g) = n$, and thus if $a_1(g) = t$, we have,
\begin{align*}
	\log(i\cdot \log a_i) ~&\leq~ \log (n-t), \quad\forall i \geq 2, \\
	\sum_{i\geq 2} a_i ~&\leq~ \frac{1}{2}\cdot  \sum_{i=2}^n ~=~ \frac{n-t}{2} \;. 
\end{align*}
Plugging this back into our computation above, we get,
\begin{align*}
	\frac{\trho_k}{(n!)^k} ~&\leq~ 1 + \sum_{t=0}^{n-2} \parens[\Big]{\frac{t!}{n!}}^k \sum_{g: a_1(g) = t}   {  e^{k\cdot \log(n-t)\sum_{i \geq 2}a_i(g)}\,}\\ 
~&\leq~ 1 + \sum_{t=0}^{n-2} \parens[\Big]{\frac{t!}{n!}}^k  \cdot e^{k\cdot \log(n-t)\frac{n-t}{2}}\sum_{g: a_1(g) = t}   { 1 \,}\\
~&\leq~ 1 + \sum_{t=0}^{n-2} \parens[\Big]{\frac{t!}{n!}}^k  \cdot e^{k\cdot \log(n-t)\frac{n-t}{2}}\cdot {n \choose t} (n-t)!\\
~&=~ 1 + \sum_{t=0}^{n-2} \parens[\Big]{\frac{t!}{n!}}^{k-1}  \cdot e^{k\cdot \log(n-t)\frac{n-t}{2}}\\
~&\leq~ 1 + \sum_{r=2}^{n} { {n \choose r}}^{1-k}  \cdot \parens{ r^{\frac{r}{2}} \cdot (r!)^{-1} }^k  &&[r=n-t]\\
~&\leq~ 1 + \sum_{r=2}^{n} { {n \choose r}}^{1-k}    &&[\text{For}\; r\geq 2, \; r^{\frac{r}{2}}  \leq r!]\\
~&\leq~ 1 + O(n^{-2(k-1)}) \;\;. \qedhere
%	~&\leq~ 1 + \sum_{t=0}^{n-2} {n \choose t} (n-t)!  \parens[\Big]{ e^{\log (n-t)\sum_{i >2} a_i}}^k(n!)^{-k}\\
%	~&\leq~ 1 + \sum_{t=0}^{n-2} {n \choose t} (n-t)!  \parens[\Big]{ e^{\frac{n-t}{2}\log (n-t)}}^k(n!)^{-k}\\
%	~&\leq~ 1 + \sum_{r=2}^{n} {n \choose r} \cdot r! \cdot  \parens[\Big]{ e^{\frac{r}{2}\log (r)}}^k(n!)^{-k}\\
%	~&\leq~ 1 + 2 n! \cdot  \parens[\Big]{ e^{\frac{n}{2}\log (n)}}^k(n!)^{-k}\\
%	~&\leq~ 1 + 2(n!)^{1-\frac{k}{2}} \;\; .
\end{align*} 
\end{proof}

Since conjugation preserves odd/even parity of the permutation, a conjugacy class of $\Sym_n$ is either entirely even or entirely even. The even classes of $\Sym_n$ either remain a conjugacy class for $A_n$, or they split into two conjugacy classes \footnote{The condition is that a class splits if and only if its cycle type consists of distinct odd integers.}. In either case, the above asymptotic analysis goes through in the same manner as above, giving us the above result for $A_n$ as well.

\subsubsection{Quasirandom Groups}

Quasirandom groups, introduced by Gowers\cite{Gow08}, is a family of \enquote{highly non-abelian} groups that is often studied in the pseudorandomness literature. It is a quantitative notion wherein we say that a group is  $D$-quasirandom if the smallest (non-trivial) irreducible representation has dimension $D$. Abelian groups are $1$-quasirandom, whereas on the other extreme, matrix groups, such as $\mathrm{PSL}_2(q)$ are $|G|^{\frac{1}{3}}$-quasirandom.

We avoid defining the relevant representation theory definitions as we will only need the following consequence of $D$-quasirandomness: every proper subgroup has size at most $|G|/D$. We sketch the derivation of this consequence below. The reader unfamiliar with representation theory can take this consequence as the definition. 

\begin{fact}\label{fact:quasirandom_index}
If $G$ is $D$-quasirandom, and $H \subseteq G$ is a non-trivial	subgroup of $H$, then, $|H| \leq \frac{|G|}{D}$.
\end{fact}
\begin{proof}
	The quasiregular representation $L^2(G/H)$ is a vector space spanned by cosets of $H$. The action of $G$ is given by group multiplication that permutes the cosets. The dimension of this representation is $\abs{G/H} = |G|/|H|$. This representation can be trivial if and only if $H = G$. Therefore, this representation contains an irreducible representation of dimension at most $ |G|/|H|$. But by quasirandomness, no irreducible representation has dimension $< D$. Hence, $|G|/|H| \geq D$. 
\end{proof}

Using the above bound on sizes of subgroups, we can easily obtain a bound on $\trho_k(G)$.

\begin{claim}\label{claim:quasi}
Let $G$ be a $|G|^c$ quasirandom group. Then, 
\[1 ~\leq~ \frac{\trho_k}{|G|^k} ~\leq~ 1 + \frac{|G|^{1-kc}}{Z(G)} \;. \]
 \end{claim}
\begin{proof}
The only observation is that $\fix(\phi_g)$ is a subgroup of $G$. Moreover, it is a proper group if and only if $g \not\in Z(G)$. This is because if $g \not \in Z(G)$, then there exists an $x$ such that $gx\neq xg$,  and thus, $x \not\in \fix(\phi_g)$.  The bound then follows: 
	\begin{align*}
		|Z(G)|\cdot \trho_k ~&=~ \sum_g \abs{\fix(\phi_g)}^k &&[\text{Definition}]\\
		~&=~ |G|^k\cdot |Z(G)| +   \sum_{g\not \in  Z(G)} \abs{\fix(\phi_g)}^k && [\text{\cref{fact:centralizer}}]\\
				\trho_k ~&\leq~ |G|^k + \frac{|G|}{|Z(G)|}\cdot \parens[\bigg]{\frac{|G|}{|G|^c} }^k \; .&& [\text{\cref{fact:quasirandom_index}}]	\qedhere
\end{align*}
\end{proof}
One family of groups that has such a large quasirandomness factor is the finite simple groups of bounded rank. This is a theorem due to  Landazuri and Seitz~\cite{LS74}, and from their result, we also extract an explicit constant for the three groups not covered by~\cref{theo:liebeck}. 

\begin{theorem}\textup{\cite[Theorem 1]{LS74}}\label{theo:landazuri}
	Every finite simple group of Lie type of rank $r$, is $G^{c(r)}$-quasirandom, where $c(r)$ is a constant only depending on $r$. In particular, for $G$ being any one of, $\mathrm{PSL}_2(q), \mathrm{PSL}_3(q),\mathrm{PSU}_3(q)$, the group $G$ is $\Theta(|G|^{\frac{1}{4}})$-quasirandom.  
\end{theorem}

Combined with~\cref{claim:quasi}, we obtain a more transparent proof of~\cref{theo:liebeck}, albeit with a weaker constant.

\subsubsection{Extraspecial Groups}

\textit{Extraspecial $p$-groups} generalize the \textit{Heisenberg group}, the group of $3\times 3$ unitriangular matrices (upper-triangular with $1$s on the diagonal) over $\F_p$. Such groups play an important role in quantum complexity. For instance, this family of groups has been studied in the context of the \textit{hidden subgroup problem}~\cite{ISS07} (see also the references within). They also appear in the context of quantum gates construction~\cite{rowell10}.    

\begin{definition}[Extraspecial group]
The \textit{Frattini subgroup}, $\Phi(G)$, of a group $G$ is the intersection of all maximal subgroups of $G$. A $p$-group is \textit{extraspecial} if $Z(G) = \Phi(G) = [G,G]$ and $\abs{Z(G)} = p$.	
\end{definition}

\begin{theorem}{\textup{\cite[Proposition 7.1]{Pan04}}}
	Let $G$ be a group of order $p^r$ such that $\abs{Z(G)} = \abs{[G,G]} = p$.  Then, $G$ has $p$ conjugacy classes of size one, and $p^{r-1} -1$ conjugacy classes of size $p$ each. And these are all the conjugacy classes. In particular, this holds for extraspecial $p$-groups.
\end{theorem}
\begin{corollary}\label{cor:extraspecial}
	For an extraspecial group $G$ of order $p^r$, $\frac{\trho_k(G)}{|G|^k} ~=~ 1 +(p^{r-1}-1)p^{-k}$.
\end{corollary}
\begin{proof}
	The result follows by plugging the sizes of conjugacy classes in \cref{cor:trho_expression}.
\end{proof}
We are now ready to prove our testing result.  
%\paragraph{Proof of~\cref{theo:main_aut}}

\begin{theorem}[Restatement of Theorem 1.4]\label{theo:main_aut2}
	The following results hold using \hyperref[test:inner]{$\test\_ {\sf Inner}_k$ } \textup{:} 
	\begin{itemize}
		\item For any $n \neq 6$, $\Aut(\Sym_n)\,$ is $\parens[\Big]{k,\,\delta,\,  \delta^{\frac{1}{k-2}} - o_n(1)}$-testable for every $k \geq 3$.
		\item For every non-abelian finite simple group $G$, $\Inn(G)$ is $\parens[\Big]{k,\,\delta,\,  \delta^{\frac{1}{k-5}} - o_{|G|}(1)}$-testable for every $k \geq 6$.
		\item For any exstraspecial group $G$ of order $p^r$, $\Inn(G)$ is $\parens[\Big]{k,\,\delta,\,  \delta^{\frac{1}{k-r}} - o_{p}(1)}$-testable for every $k \geq r+1$. In particular, for the family of Heisenberg groups over $\F_p$, $H_p$, $\Inn(H_p)$ is   $\parens[\Big]{k,\,\delta,\,  \delta^{\frac{1}{k-3}} - o_{p}(1)}$-testable for any $k \geq 4$.
		\item Alternatively, if $p = O(1)$ is fixed and $r \to \infty$, then we have that $\Inn(G)$ is $\parens[\Big]{k,\,\delta,\, \frac{1}{p} \delta^{\frac{1}{k-1}}}$-testable for every $k \geq 2$, and $r\geq 3$.
	\end{itemize}
	Additionally, we  get an upper bound of $\max_\phi \agrHom \leq 2\delta^{\frac{1}{k}}$, for any group $G$ and $k$ as above.
	\end{theorem}
\begin{proof}To use \cref{lem:geninn_test}, all we need to do is bound $\frac{\trho_i}{G| \trho_{i-1}}$ for all $i \geq \tau$, for some $\tau$.
For the symmetric group, \cref{lem:rho_sym} gives a bound on $\trho_k$ for all $k \geq 2$, and thus, we have a bound on the ratio for all $k \geq 3$. For finite simple groups, \cref{claim:quasi} coupled with \cref{theo:liebeck,theo:landazuri}, gives a bound of $\frac{\trho_k}{|G|\trho_{k-1}} \geq 1-o_{|G|}(1)$ for every $k \geq 6$. For the extraspecial groups, we use~\cref{cor:extraspecial}.
\end{proof}

\section{Lifting Homomorphism Tests}\label{sec:beyond}

\subsection{A General Lifting Lemma}\label{subsec:lifting}

In this section, we will see how to lift the analysis of our test for $\Hom(G,H)$ to that of $\Hom(\tG,\tH)$. A key point is that this lifting is not an algorithmic reduction but a method to reuse our analysis by utilizing the fact that it only depends on group-theoretic constants.

% which separates the function-dependent component ($\delta_k(f)$) 
%Let $G, H$ be abelian groups such that $\test(G,H,\cD)$ works 
%
% and $\pi_G: \tG \to G$ be a surjective homomorphism where $\tG$ is not necessarily abelian. We will now show that if $\Hom(\tG,H) \cong \Hom(G,H)$ via this surjection, \ie 

\begin{definition}[Lifted Homomorphisms]
Let $\tG, \tH, G, H$ be finite groups such that there is a surjective homomorphism $\pi_G: \tG \to G$, and an injective inclusion $\iota_H: H\to \tH$. Then we define the subset of \textit{lifted homorphisms} as those obtained by lifting $\Hom(G,H)$,
\[
\LiftHom(\tG, \tH) = \braces[\big]{\iota_H\circ \phi\circ \pi_G \mid \phi \in \Hom(G,H) } \subseteq \Hom(\tG,\tH).
\]
\[
\begin{tikzcd}
 \tilde{G} \arrow[d, two heads, "\pi_G"] \arrow[r, dashed, " \tilde{\varphi}"] & \tilde{H} \\
 G \arrow[r,"\varphi" below] & H\arrow[u, hook, "\iota_H"]
\end{tikzcd}
\]	
 Note that if $\Hom(G,H)$ is an abelian group, then so is $\LiftHom(\tG,\tH)$.
\end{definition}

Let us now see a simple natural example in which the set of lifted homomorphisms contains all possible homomorphisms. We will use this and similar examples later to derive our more general results. 

\begin{example}\label{example:lift_cyclic}
Let $G = H = \tH = \Z_{p^r}$. And pick $\tG = \Z_{p^s}$ for $s >r$. Then, there is a natural projection $\tG\to G$ whose kernel (apart from $0$) is precisely the set of elements of order greater than $p^r$. Now, for any homomorphism from $\phi: \tG\to H$, $\phi(x) = x\phi(1)$ but $\phi(1)\in \Z_{p^r}$ has order $\leq p^r$. Thus, $\ker (\pi) \subseteq \ker(\phi)$. This shows that every homomorphism factors through $G$ and thus, $\Hom(\tG, H) = \LiftHom(\tG, H)$.    
\end{example}

\paragraph{The lifted test} Our general machinery works identically as  $\widetilde{\Gamma}_{\vec{x}}: \LiftHom(\tG,\tH) \to \tH^k$ is still a map between groups and has the $N$-to-one property. The group $\tH_{\vec{x}} = \mathrm{Im}\parens{\widetilde{\Gamma}_{\vec{x}}}$ is defined as before. 

%For a distribution $\cD$ on $G^k$, one can define the lifted distribution $\pi_G^{-1} \cD$.    
%\begin{tcolorbox}[colframe=teal, colback=white, title=$\gentest_k$]
%	\begin{itemize}
%		\item Sample $(x_1,\dots, x_k) \sim \cD_k$.
%		\item If $(f(x_1),\dots, f(x_k))\in \sfH_{\vec x}$: return $1$; otherwise: return $0$
%	\end{itemize}

\begin{lemma}\label{lem:lift_gamma}
	Let $G, H, \tG, \tH$ be groups as above, and let $k \geq 1$ be an integer. 
	\begin{align*}
		\abs{\ker(\widetilde{\Gamma}_{\vec x})} ~&=~ \abs[\big]{\ker(\Gamma_{\pi_G({\vec x})})},\\
		\gamma_k(\tG, \tH) ~&=~ \abs{\ker{\pi_G}}^k\cdot  \gamma_{k}(G,H),\\
		\eta_k(\tG, \tH) ~&=~ \eta_k(G, H) \qquad \text{ if }\; \iota_H \;\;\text{is an isomorphism},\\
					~&=~    1 \qquad \text{otherwise}. 
	\end{align*}
	Moreover, for any subset $S\subseteq G^k$, $\Pr{\vec y \in G^k}{\vec y \in S} = \Pr{\vec x \in \tG^k}{ \pi(\vec x) \in S}$.
\end{lemma}
\begin{proof} All the claim follow simply by writing out the definitions. 
\begin{align*}
	\ker(\widetilde{\Gamma}_{\vec x}) ~&=~ \{\iota\circ \phi\circ \pi_G \mid \iota\circ \phi\circ \pi_G (\vec x) = 0  \}\;,\\
~&\simeq~ \{\phi \mid \phi\circ \pi_G (\vec x) = 0  \}\;,\\	
~&=~ \ker(\widetilde{\Gamma}_{\pi(\vec x)}) \;\;.	
\end{align*}
Observe that since $\iota_H$ is injective, $\ker(\iota_H\circ\phi\circ \pi_G) = \pi_G^{-1} \parens{\ker(\phi)}$. Since, $\pi_G$ is a surjection we have that $\pi_G$ is a $|\ker{\pi_G}|$ to one map on the entire image which is $G$. Thus,
% $\abs{\pi_G^{-1} \parens{\ker(\phi)}} = \abs{\ker(\pi_G)} \abs{\ker(\phi)}$
\[
\abs{\ker(\iota_H\circ\phi\circ \pi_G)} ~=~ \abs{\pi_G^{-1} \parens{\ker(\phi)}} ~=~ \abs{\ker(\pi_G)} \cdot \abs{\ker(\phi)} \;.
\]
The computation of $\gamma_k$ is similar.
\begin{align*}
	\gamma_k(\tG, \tH) ~&=~ \sum_{\psi \in \LiftHom(\tG, \tH)} \abs{\ker{\psi}}^k\;,\\
	~&=~ \sum_{\phi \in \Hom(G, H)} \abs{\ker{\iota_H\circ\phi\circ \pi_G}}^k\;,\\
		~&=~ \abs{\ker(\pi_G)}^k \sum_{\phi \in \Hom(G, H)} \abs{\ker{\phi}}^k ~=~ \abs{\ker(\pi_G)}^k\cdot \gamma_{k}(G,H).
\end{align*}
Now, by definition $\eta_k$ is the fraction of tuples $\vec{x}$ such that $\widetilde{\sfH}_{\vec{x}} \neq \tH^k$. But, $\widetilde{\sfH}_{\vec{x}} = \iota_H\parens{{\sfH}_{\pi(\vec{x})}}$ 
Clearly, if $\iota_H$ is not an isomorphism,  $\eta_k = 1$ as $\widetilde{\sfH}_{\vec x} \subsetneq \tH^k$ for any $k \geq 1$. If it is, then it is equal precisely to the fraction of tuples for which ${\sfH}_{\vec{x}} = H^k$, \ie for an $\eta_k(G,H)$-fraction. The last claim is a simple consequence of the $N$-to-one property of $\pi$.
\end{proof}

In the following few subsections, we will utilize this lifting technique to extend the results we have obtained to other settings.

\subsection{Character Testing via Abelianization trick}\label{subsec:character}

When the target $H$ is abelian, all homomorphisms from $G\to H$, factorize through the \textit{abelianistaion} of $G$, and thus all the homomorphisms are in fact lifts of homomorphisms between abelian groups. We state this more precisely: 

\begin{fact}[Homomorphisms abelianize]\label{fact:hom_ab}
	Let $G$ be any group and $H$ be an abelian group. Let $[G,G] = \angles{ ghg^{-1}h^{-1} \mid g, h\in G}$ be the derived (or commutator) group, and $G/[G,G]$ its abelianisation. Then, $\Hom(G, H) \cong \Hom(G/[G,G], H)$.
\end{fact}

Thus, $\Hom(G,H) = \LiftHom(G,H)$ where the lift is via the  projection $\pi_G: G\to [G,G]$.

\subsubsection{Character Testing for $\GL_n(q)$}
Let $\GL_n(q)$ be the group of inverible $n\times n$ matrices over $\F_q$ for $q >2, n\geq 2$, and let $\F_q^*$ be the multiplicative group of non-zero elements of $\F_q$. Note that $\F_q^* \cong \Z_{q-1}$.  We wish to study the testing of homomorphisms $f: \GL_n(q)\to \F_q^*$, \ie the $\F_q$-characters or the one-dimensional \textit{representations}.

\begin{fact}[Linear Characters of $\GL_n(q)$]\label{fact:characters_gln}
	Let $G = \GL_n(q)$ for $q >2$. Then, $[G,G] = \SL_n(q)$, \ie matrices of determinant $1$, and thus, $G/[G,G] \cong \F_q^* \cong \Z_{q-1}$. Thus, $\Hom(\GL_n(q), \F_q^*) \cong \Hom(\F_q^*, \F_q^*)$. Moreover, for any prime power $q$, $\Z_{q-1}$ is a cyclic group.
\end{fact}

We are now ready to use the lifting machinery and deduce a result for testing characters of $\GL_n(q)$, \ie $\Hom(\GL_n(q), \F_q^*)$ as we already have a testing algorithm~(\cref{thm:cyclic_gen}) for $\Hom(\F_q^*,\F_q^*) = \Hom(\Z_{q-1},\Z_{q-1})$.
% Note that this reduction is not an algorithmic reduction but a reduction in our analysis which separates out the function-dependent component ($\delta_k$) from the group theoretic computation $\gamma_k(G, H)$.

\begin{theorem}[Character Testing for $\GL_n(q)$]\label{theo:GL_n}
Let $G = \GL_n(q)$ be the group of inverible $n\times n$ matrices over $\F_q$ for $q >2, n\geq 2$. Let $f: G\to \F_q^*$ be any function and fix an integer $k\geq 4$. Then if $f$ passes $\test_k$ with probability $\delta_k$, there exists a character $\phi \in \Hom(\GL_n(q), \F_q^*)$ such that $\agr(f,\phi) ~\geq~  \parens[\big]{\zeta(2)^2\cdot\delta_k}^{\frac{1}{k-3}}$.
\end{theorem}
\begin{proof}
Let $\det: \GL_n(q)\to \F_q^*$ be the surjection that maps $\GL_n(q)\to [\GL_n(q),\GL_n(q)]$. then, $|\ker(\det)| = \frac{\abs{\GL_n(q)}}{q-1}$. 	Then using from~\cref{lem:lift_gamma} we get that $\eta_k(\GL_n(q),\F_q^*) =1$ and thus, using~\cref{lem:eta_one}:
	\begin{align*}
		\sum_{\phi\in\Hom(\GL_n(q), \F_q^*)}\agrHom^k ~&=~ \delta_i \cdot\frac{\gamma_i(\GL_n(q), \F_q^*)}{\abs{\GL_n(q)}^k}\\
		~&=~ \delta_i\cdot \frac{|\ker(\det)|^k\cdot \gamma_i(\F_q^*, \F_q^*)}{\abs{\GL_n(q)}^k} &&[\text{\cref{lem:lift_gamma}}]\\
		~&=~ \delta_i \cdot  \frac{\gamma_i(\F_q^*, \F_q^*)}{(q-1)^k} .
	\end{align*}
	Now, one can reuse the proof of~\cref{thm:cyclic_gen}, as the RHS expression is identical. 
\end{proof}

\subsubsection{Character Testing for Lie Algebras}

The above abelianization approach generalizes to other structures such as \textit{Lie algebras}. A finite-dimensional Lie algebra, $\mathfrak{g}$, is a finite-dimensional vector space over a field $\F$ with a Lie bracket $[\cdot, \cdot ] : \fkg\times \fkg\to \fkg$ which is a bilinear map such that 
\[
 [x,x] = 0 \quad \text{and} \quad  [x,[y,z]] + [y,[z,x]] + [z,[x,y]] = 0, \;\; \forall \,  x,y,z \in \fkg.
\]
% \begin{enumerate}
% 	\item $[x,x] = 0$, and
% 	\item  $[x,[y,z]] + [y,[z,x]] + [z,[x,y]] = 0 $, for any $x,y,z \in \fkg$,
% \end{enumerate}
A Lie algebra is \textit{abelian} if $[\cdot, \cdot ]$ is identically zero. 

\begin{definition}[Lie algebra homomorphisms] A map $\phi: \fkg\to \fkh$ between two Lie algebras is a homomorphism if it is a linear map, \ie a homomorphism as vector spaces, and additionally if $[f(x), f(y)] = f([x,y])$. We refer to these homomorphisms as $\LieHom(\fkg, \fkh)$. 
\end{definition}

As an example, let $\mathfrak{gl}_n(q) = \F_q^{n\times n}$, the space of $n\times n$-matrices. It clearly has a vector space structure, and we define the bracket as $[x,y] := xy-yx$, where the multiplication is matrix multiplication. Thus, for any abelian $\fkh$, a Lie algebra homomorphism is a linear map such that $f([x,y]) = f(xy) -f(yx) = 0$ for every $x,y \in \fkg$. We will now see that this reduces to lifted homomorphisms between vector spaces.

\begin{fact}[Lie algebra homomorphisms abelianize]\label{fact:lie_ab}
	Let $\fkg$ be a finite-dimensional Lie algebra and define its derived algebra as $[\fkg,\fkg] = \mathrm{span}\braces{ [x,y] \mid x, y\in \fkg}$. Then, for any abelian Lie algebra $\fkh$, $\LieHom(G, H) \cong \Hom(\fkg/[\fkg,\fkg], \fkh) \cong \Hom(\F_q^n, \F_q^m)$ where the $n,m$ are the dimensions of $\fkg/[\fkg,\fkg]$ and $\fkh$ as a $\F_q$-vector space, also called the ranks of $\fkg, \fkh$.
\end{fact}

Again for our example, $\fkg = \mathfrak{gl}_n(\F_q)$, we have $\fkg/[\fkg,\fkg] \cong \F_q$. This is because $\tr([x,y]) = 0$ for every $x,y$. Moreover, every trace $0$ matrix can be generated by a linear span of such commutators\footnote{To show this, one can use $[e_{ij}, e_{k\ell}] = \delta_{jk}e_{i\ell} - \delta_{i\ell}e_{kj}$. Here $e_{i,j}$ represents the elementary matrix which is zero everywhere and has $1$ in the $(i,j)$ entry.}, and thus, $[\fkg,\fkg]$ is the algebra of trace-zero matrices. We therefore have the folowing isomorphism for its characters, \ie \[\mathrm{LieHom}(\mathfrak{gl}_n(\F_q), \F_q) ~\cong~ \Hom(\F_q, \F_q).\] More explicity, the characters are of the form $\chi\circ \tr$ for any homomorphism $\chi: \F_q\to \F_q$ while for $\GL_n(q)$ they were $\chi \circ \det$ for $\chi: \F_q^*\to \F_q^*$. 
Now, we can deduce a testing result by lifting the cyclic case, just as we did for $\GL_n(q)$. The only difference is that we use~\cref{thm:cyclic_bounded} instead of~\cref{thm:cyclic_gen} to obtain a slightly better query complexity.

%\begin{theorem}[Testing prime power cyclic groups to Abelian groups of bounded rank]\label{thm:cyclic_bounded}
%		Let  $G=\Z_{p^r}$ be a cyclic group and $H$ be an abelian group of $p$-rank $t\geq 1$. Let $k \geq t+2$ be an integer, and define $\cD_k$ as $\vec{x}$.  Let $f: G\to H$ be any function. Then if $f$ passes $\test_k(G,H, \cD_k)$ with probability $\delta_k(f)$, then there exists a homomorphism $\phi \in \Hom(G,H)$ such that $\agr(f,\phi) ~\geq~  \parens[\Big]{\frac{(p-1)^2}{p^2}\cdot \delta_k(f)}^{\frac{1}{k-t}} .$ 
%\end{theorem}

\begin{theorem}[Character Testing for $\mathfrak{gl}_n(q)$]\label{theo:gln}
	Let $q$ be a power of a prime $p$, and let $\fkg =  \mathfrak{gl}_n(q) = \F_q^{n\times n}$, be the space of $n\times n$-matrices. Let $k \geq 3$ be an integer.  Let $f: \fkg\to \F_q$ be any function. If $f$ passes $\test_k(G,H, \cD_k)$ with probability $\delta_k$, then, 
%	\vspace{-0.5em}
	\[
\max_{\phi \in \Hom(\fkg,\F_q)} \agr(f,\phi) ~\geq~  \parens[\Big]{\frac{(p-1)^2}{p^2}\cdot \delta_k}^{\frac{1}{k-1}} .\]
\vspace{0.1em}
\end{theorem}
\begin{proof}
Let $\tr: \mathfrak{gl}_n(q) \to \F_q$ be the surjection that maps $\mathfrak{gl}_n(q) \to [\mathfrak{gl}_n(q) ,\mathfrak{gl}_n(q) ]$. then, $|\ker(\tr)| = \frac{\abs{\mathfrak{gl}_n(q) }}{q}$ 	Then using from~\cref{lem:lift_gamma} we get that $\eta_k(\mathfrak{gl}_n(q),\F_q) =1$ and thus, we use~\cref{lem:eta_one}:
	\begin{align*}
		\sum_{\phi\in\Hom(\mathfrak{gl}_n(q) , \F_q^*)}\agrHom^k ~&=~ \delta_i \cdot\frac{\gamma_i(\mathfrak{gl}_n(q) , \F_q^*)}{\abs{\mathfrak{gl}_n(q)}^k}\\
		~&=~ \delta_i\cdot \frac{|\ker(\tr)|^k\cdot \gamma_i(\F_q^*, \F_q^*)}{\abs{\mathfrak{gl}_n(q)}^k} &&[\text{\cref{lem:lift_gamma}}]\\
		~&=~ \delta_i \cdot  \frac{\gamma_i(\F_q, \F_q)}{q^k} .
	\end{align*}
	Now, one can reuse the proof of~\cref{thm:cyclic_bounded}, as the RHS expression is identical. 
\end{proof}

To extend this result to arbitrary Lie algebras $\fkg$, we will need to lift the result for the vector space, \ie $\Hom(\F_q^n, \F_q)$. Since that proof does not directly follow from the computation of $\gamma$, we will need to derive it a bit more carefully which we do now.  

%\paragraph{Defining the test}
%
%%\begin{figure}
%
%{Lie algebra homomoprhisms of $\fkg$ are precisely those vector space homomorphisms that are lifted from $\Hom(\fkg/[\fkg,\fkg] , \fkh)$}
%\[
%\begin{tikzcd}
% \fkg   \cong \F_q^N  \arrow[d, two heads, "\pi" left] \arrow[rd, dashed, "\forall \tilde{\varphi}"]  \\
% \fkg/[\fkg,\fkg] \cong \F_q^n \arrow[r,"\exists!\varphi" below] & \fkh \cong  \F_q^m
%\end{tikzcd}
%\]
%\end{figure}

%\begin{theorem}[Lifting]
%	Let $G,\tG, H, \tH$ be as in the lemma above. Then if
%	 $\test_k(G,H, \cD)$ is a $(\delta, \eps)$-test for $\Hom(G,H)$, then $\test_k(\tG,\tH, \pi_G^{-1}\cD)$ a $(\delta, \eps)$-test for $\LiftHom(\tG,\tH, )$.
%\end{theorem}
%\begin{proof}
%	
%\end{proof}

%\begin{theorem}[Lifting]
%	Let $G,\tG, H, \tH$ be as in the lemma above. If
%	$\Hom(G,H)$ is $(\delta, \eps)$-testable, then $\Hom(G,H)$ is $(\delta, \eps)$-testable.
%%	 $\test_k(\tG,\tH)$ a $(\delta, \eps)$-test for $\Hom(\tG,\tH)$.
%\end{theorem}

\subsection{Lifting Vector Space} 

We will now specialize to the case when $G,H = (\F_q^n, \F_q)$. Let $\tG$ be a group that projects to $\F_q^n$, and we will take $\tH = H$. The goal is to deduce testing for $\Hom(\tG, \F_q)$ from that for $\Hom(\F_q^n, \F_q)$ in~\cref{sec:vspace}.  For $\vec{x} \in \tG^k$, define it $\rank(\vec x) = \rank(\pi(\vec x))$. Moreover, from~\cref{lem:lift_gamma} for any $j$, $\Pr{{\vec x} \in G^k}{\rank(\vec x) = j} = \Pr{{\vec y} \in (\F_q^n)^k}{\rank(\vec y) = j}  $. 

\begin{claim}[{A lifted variant of~\cref{claim: Vexpression}}]\label{claim: Vexpression2}
	 Let $\tG$ be any group such that $\pi:\tG\to \F_q^n$ is a surjection, $k\geq 1$ be any integer. Then, 
	 \[\sum_{\phi \in \LiftHom(\tG, \F_q)} \agrHom^k ~\approx_{q^{-n}}~ \Pr{\vec x\sim (\F_q^n)^k}{\rank(\vec x)=k}  \cdot q^{n-k} + q^{-(k-1)} \cdot    \sum^{k}_{j=1} \binom{k}{j} (q-1)^{j-1} \delta_j'(f)\;, \]
	 	where $\delta_j'(f) ~:=~  \Ex{\vec x\sim G^j}{\vtest \mid \pi(\vec x)\in {\cal R}_j}.$ where ${\cal R}_j$ is as in~\cref{def:rj}.
\end{claim}

\begin{proof}
We restate~\cref{eq:t1t2} in~\cref{claim: Vexpression} for any $g:\F_q^n\to \F_q$,
\[
		\sum_{\phi \in \Hom(\F_q^n, \F_q) } \agr(g,\phi)^k ~\approx_{q^{-n}}~   \sfT_{k-1} +  \sfT_{k}
\]
Denote these terms as $\sfT_k(\F_q^n)$ and $\sfT_{k-1}(\F_q^n)$. We will now compute these two terms for an arbitrary $\tG$ that projects to $\F_q^n$. For any $\vec x\in \tG^k$, define 
	\begin{align}
\nonumber	\beta(\vec x) ~&:=~ \indicator{f(\vec x) \in \Hx} \cdot \abs[\big]{\ker(\widetilde{\Gamma}_{\vec x}) } \\[1.5pt]
\nonumber~&=~ \indicator{f(\vec x) \in \Hx} \cdot \abs[\big]{\ker(\Gamma_{\pi({\vec x})}) } && [\text{\cref{lem:lift_gamma}}] \\
\label{eq:beta2}	~&=~ q^{n-\rank(\vec x)} \cdot \indicator{f(\vec x) \in \Hx} \;\;. 
	\end{align}

If $\rank(\vec x)=k$, then $\sfH_{\vec x} =  \sfH_{\pi(\vec x)}=H^k$, and thus, $\indicator{f(\vec x) \in \Hx}=1$ for all such full rank $\vec x$. 
	This gives us that $\sfT_k$ term is unchanged.
\begin{align*}
\sfT_{k}(G) ~&:=~ \Pr{\vec x\sim \tG^k}{\rank(\vec x)=k} \cdot q^{n-k} ~=~ \Pr{\vec y\sim  (\F_q^n)^k}{\rank(\vec y)=k} \cdot q^{n-k} = T_k(\F_q^n). 
%~&:=~ \Pr{\vec x\in (\F_q^n)^k}{ \rank(\vec x)=k }\cdot  \Ex{\vec x\in G^k}{ \beta(\vec x) \mid  \rank(\vec x)=k} = \\
\end{align*} 
	
We now compute $ \sfT_{k-1}$, exactly as before . The only change is that the expression for $\delta_j$ is different but its coefficients, which are constants depending on the group are exactly the same as that for $\F_q$ due to the fact that projection $\pi$ is $N$-to-one. 
%we further partition the set of rank $k-1$ tuples by the linear relation they satisfy. Observe that any such tuple, $\vec x$, must satisfy exactly one relation. Thus each such tuple belongs to ${\cal R}_j$ for some $j\in \{1,\dots k \}$ and this is the partitioning we use.
	%Moroever, the relation must be of the form
	%$\sum_{i\in S}c_ix_i$ for some $\emptyset \neq S\subseteq [i]$ and  $c_i\in \F^*_{q}$ for $i\in S$ with $c_{\max\{S\}}=1$. 
	\begin{align*}
		\sfT_{k-1} (f) &~=~	\sum^{k}_{j=1}   \Pr{\vec x\in \tG^k}{ \pi(\vec x) \in {\cal R}_{j} }		\cdot \Ex{\vec x\in \tG^k}{ \beta(\vec x) \mid  \pi(\vec x) \in {\cal R}_{j} }\; , \\
		&~=~   q^{n-(k-1)} \cdot    \sum^{k}_{j=1}   \Pr{\vec x\in \tG^k}{ \pi(\vec x) \in {\cal R}_{j} }  \cdot  
		 \Ex{\vec x\in \tG^k}{ \vtest \mid  \pi(\vec x) \in {\cal R}_{j} } &&[\text{Using~\cref{eq:beta2}}] \; ,
		\\
		&~=~   q^{n-(k-1)} \cdot    \sum^{k}_{j=1}   \Pr{\vec y\in (\F_q^n)^k}{ \vec y \in {\cal R}_{j} }  \cdot  
		 \delta_j'(f) \;\;.
	\end{align*}
	
	Therefore, we have obtained an identical expression as in ~\cref{claim: Vexpression} except that $\delta_j$ is replaced by $\delta_j'$. The claim then follows from that computation.		
%	The claim follows. 
\end{proof}

We can now easily define the lifted test as:
\begin{tcolorbox}[colframe=teal, colback=white, title={${\test\_\mathsf{LiftedVSpace}}_k(f)$}, label=test:lifted]
	\begin{itemize}
		\item Sample $(x_1,\dots, x_k) \sim \pi^{-1}\parens{{\cal R}_k}$.
		\item If $(f(x_1),\dots, f(x_k))\in {\sfH}_{\vec x}$: return $1$; otherwise: return $0$
	\end{itemize}
	%$D_{\ker}(\vec x)=\frac{|\ker(\Gamma_{\vec x})|}{         \sum_{{\vec x \in G^k~:~ {\sf H}_{\vec x}\neq H^k}}   |\ker(\Gamma_{\vec x})| }$
\end{tcolorbox}

\begin{theorem}[{Lifted variant of~\cref{theo:vspace}}]\label{theo:vspace_lift}
 Let  $\tG$ be any group that projects to $\F^n_{q}$ for some finite field of order $q$. Let $k\geq 3$ be any odd integer. Then if $f$ passes  \hyperref[test:lifted]{$\test\_\mathsf{LiftedVSpace}_k$}
  with probability $\delta_k$, there exists a homomorphism $\phi \in \LiftHom(\tG,\F_q)$ such that:
 \[\agr(f,\phi) ~\geq~ \frac{1}{q} + \parens[\Big]{\frac{q-1}{q}} \parens[\Big]{\frac{q\delta_k-1}{q-1}}^{\frac{1}{k-2}} \;\;.\]
\end{theorem}
\begin{proof}
	The proof is identical to the short computation in~\cref{theo:vspace} as~\cref{cor:Vexptest} holds just as before by using~\cref{claim: Vexpression2} instead of~\cref{claim: Vexpression}.
\end{proof}

As an application of the above theorem, we can generalize the results to two setups.

\begin{corollary}\label{cor:lift_lie}
For any $k\geq 3$, \hyperref[test:lifted]{$\test\_\mathsf{LiftedVSpace}_k$} is a $(k,\delta, \ep(\delta))$ sound test for $\Hom(G,\F_p)$ for any finite group $G$ and prime $p$, and for $\LieHom(\fkg,\F_q)$ for any finite-dimensional Lie algebra, $\fkg$,  and prime power $q$.
\end{corollary}
\begin{proof}
We briefly sketch the arguments which just collect our earlier observations:  

\begin{itemize}
	\item From~\cref{fact:characters_gln} for any finite-dimensional Lie algebra, $\fkg$,  $\Hom(\fkg, \F_q) = \LiftHom(\F_q^n, \F_q)$ where $\pi : \fkg\to [\fkg,\fkg]$ is some canonical projection.
	\item Let $G$ be an arbitrary finite group and let $G/[G,G]$ have $p$-rank $n$, and let its $p$-component be $\oplus_{i=1}^n \F_{p^{b_i}}$. Then, we have a projection $\pi = \oplus_i \pi_i : G\to \F_p^n$ where $\pi_i: \F_{p^{b_i}}\to \F_p$ be the canonical projection. Combining~\cref{fact:hom_ab} and~\cref{example:lift_cyclic}, we get  $\Hom(G,\F_p) = \LiftHom(\F_p^n,\F_p)$. 
	\item If one wishes to generalize the above to a prime power $q$, we need the condition that if $q = p^a$, then the $p$-component of $G/[G,G]$ has summands of larger order than $q$. That is, if $\parens{G/[G,G]}_p = \oplus_i \F_{p^{b_i}}$, then $b_i \geq a$ for each $i$. \end{itemize}
Now, we may plug these all into~\cref{theo:vspace_lift}.
\end{proof}

%
%
%\begin{corollary}[Character Testing for Lie Algebras]
%	Let $\fkg$ be any finite-dimensional Lie algebra and let $\fkg/[\fkg,\fkg] \cong \F_q^n$. Then,  
%\end{corollary}
%
%\begin{corollary}[Character Testing for Lie Algebras]
%	Let $\fkg$ be any finite-dimensional Lie algebra and let $\fkg/[\fkg,\fkg] \cong \F_q^n$. Then,  
%\end{corollary}
%

%
% $\frac{\mathfrak{gl}_n(\F_q)}{[\mathfrak{gl}_n(\F_q), \mathfrak{gl}_n(\F_q)]} \cong \F_q$, and , and thus the Homset is a group and we can redo the same thing. 
%
%\todo{write definitions, and double check the statements here and cite the case for modular representations.}

%\section{Sampling from $\cD_{\ker}$}
%\input{old_stuff.tex}

\bibliographystyle{alphaurl}
\bibliography{macros,references}

@string{stoc90 ="Proceedings of the 22nd ACM Symposium on Theory of Computing"}

@string{stoc03 ="Proceedings of the 35th ACM Symposium on Theory of Computing"}

@string{stoc07 ="Proceedings of the 39th ACM Symposium on Theory of Computing"}

@string{stoc08 ="Proceedings of the 40th ACM Symposium on Theory of Computing"}

@string{stoc17 ="Proceedings of the 49th ACM Symposium on Theory of Computing"}

@string{stoc21 ="Proceedings of the 53rd ACM Symposium on Theory of Computing"}

@string{focs95 ="Proceedings of the 36th IEEE Symposium on Foundations of Computer Science"}

@STRING{acm 	= "Association for Computing Machinery"}

@STRING{and	= " and "}

@article{KRV23,
  title={{Sublinear-Time Computation in the Presence of Online Erasures}},
  author={Kalemaj, Iden and Raskhodnikova, Sofya and Varma, Nithin},
  journal={Theory of Computing},
  volume={19},
  number={1},
  pages={1--48},
  year={2023},
  publisher={Theory of Computing Exchange},
  eprint={2109.08745}
}

@InProceedings{BKMR24,
  author =	{Ben-Eliezer, Omri and Kelman, Esty and Meir, Uri and Raskhodnikova, Sofya},
  title =	{{Property Testing with Online Adversaries}},
  booktitle =	{15th Innovations in Theoretical Computer Science Conference (ITCS 2024)},
  year =	{2024},
  doi =		{10.4230/LIPIcs.ITCS.2024.11},
  eprint = {2311.16566}
  }

@misc{KMNR25,
      title={Homomorphism Testing with Resilience to Online Manipulations}, 
      author={Esty Kelman and Uri Meir and Debanuj Nayak and Sofya Raskhodnikova},
      year={2025},
      eprint={2511.23363}
      }

@article{rowell10,
	title={Extraspecial two-groups, generalized yang-baxter equations and braiding quantum gates},
	author={Rowell, Eric C and Zhang, Yong and Wu, Yong-Shi and Ge, Mo-Lin},
	journal={Quantum Information \& Computation},
	volume={10},
	number={7},
	pages={685--702},
	year={2010},
	publisher={Rinton Press, Incorporated Paramus, NJ}
}

@misc{MR24,
      title={{Derandomized Non-Abelian Homomorphism Testing in Low Soundness Regime}}, 
      author={Tushant Mittal and Sourya Roy},
      year={2024},
      eprint={2405.18998},
      archivePrefix={arXiv},
      primaryClass={cs.CC},
}

@article{PPR06,
	title={Tolerant property testing and distance approximation},
	author={Parnas, Michal and Ron, Dana and Rubinfeld, Ronitt},
	journal={Journal of Computer and System Sciences},
	volume={72},
	number={6},
	pages={1012--1042},
	year={2006},
	publisher={Elsevier}
}

@article{Gopalan13,
	author       = {Parikshit Gopalan},
	title        = {{A Fourier-Analytic Approach to Reed-Muller Decoding}},
	journal      = {{IEEE} Trans. Inf. Theory},
	volume       = {59},
	number       = {11},
	pages        = {7747--7760},
	year         = {2013},	
	doi = {10.1109/TIT.2013.2274007},
}

@article{LS74,
title = {On the minimal degrees of projective representations of the finite Chevalley groups},
journal = {Journal of Algebra},
volume = {32},
number = {2},
pages = {418-443},
year = {1974},
issn = {0021-8693},
doi = {https://doi.org/10.1016/0021-8693(74)90150-1},
author = {Vicente Landazuri and Gary M Seitz}
}

@article{Seg40,
 author = {Segal, Irving E.},
 title = {The automorphisms of the symmetric group},
 fjournal = {Bulletin of the American Mathematical Society},
 journal = {Bull. Am. Math. Soc.},
 issn = {0002-9904},
 volume = {46},
 pages = {565},
 year = {1940},
 language = {English},
 doi = {10.1090/S0002-9904-1940-07261-1},
 zbMATH = {3099424},
 Zbl = {0061.03301}
}

@article{LS05, 
title={Character degrees and random walks in finite groups of {Lie} type}, 
volume={90},
 DOI={10.1112/S0024611504014935},
  number={1},
   journal={Proceedings of the London Mathematical Society},
    author={Liebeck, Martin W. and Shalev, Aner},
     year={2005}, 
     pages={61–86}
     }

@inproceedings{ISS07,
author = {Ivanyos, G\'{a}bor and Sanselme, Luc and Santha, Miklos},
title = {An efficient quantum algorithm for the hidden subgroup problem in extraspecial groups},
year = {2007},
isbn = {9783540709176},
publisher = {Springer-Verlag},
address = {Berlin, Heidelberg},
booktitle = {Proceedings of the 24th Annual Conference on Theoretical Aspects of Computer Science},
pages = {586–597},
numpages = {12},
location = {Aachen, Germany},
series = {STACS'07},
eprint = {quant-ph/0701235}
}

@MISC {stackex,
    TITLE = {A general formula for the number of conjugacy classes of $\mathbb{S}_n \times \mathbb{S}_n$ acted on by $ \mathbb{S}_n$},
    AUTHOR = {Marty Isaacs},
    HOWPUBLISHED = {MathOverflow},
    year = {2014},
    NOTE = {URL:https://mathoverflow.net/q/162275 (version: 2014-04-03)},
}

@article{Pan04,
  title = {ON THE NUMBER OF CONJUGACY CLASSES OF FINITE p-groups},
  author = {Pantea, Casian Alexandru},
  langid = {english},
  year = {2004},
  journal = {Mathematica (Cluj)},
  volume = {46 (69)},
  number = {2},
  pages = {193-203},
  publisher = {EDITURA ACADEMIEI ROMA{\^A}NE} 
}

@inproceedings{GKS06,
	author       = {Elena Grigorescu and
	Swastik Kopparty and
	Madhu Sudan},
	title        = {Local Decoding and Testing for Homomorphisms},
	booktitle    = {Approximation, Randomization, and Combinatorial Optimization. Algorithms and Techniques},
	pages        = {375--385},
	publisher    = {Springer},
	year         = {2006},
	doi          = {10.1007/11830924\_35}
}

@article{MIP21,
	author = {Ji, Zhengfeng and Natarajan, Anand and Vidick, Thomas and Wright, John and Yuen, Henry},
	title = {{MIP*} = {RE}},
	year = {2021},
	issue_date = {November 2021},
	publisher = {Association for Computing Machinery},
	address = {New York, NY, USA},
	volume = {64},
	number = {11},
	issn = {0001-0782},
	doi = {10.1145/3485628},
	journal = {Commun. ACM},
	pages = {131--138},
	numpages = {8}
}

@InProceedings{NV17,
	author = {Natarajan, Anand and Vidick, Thomas},
	title = {A quantum linearity test for robustly verifying entanglement},
	year = {2017},
	isbn = {9781450345286},
	doi = {10.1145/3055399.3055468},
	booktitle = stoc17
}

@misc{GGMT23,
      title={On a conjecture of Marton}, 
      author={W. T. Gowers and Ben Green and Freddie Manners and Terence Tao},
      year={2023},
      eprint={2311.05762},
      archivePrefix={arXiv},
      primaryClass={math.NT}
}

@InProceedings{GS14,
  author    = {Guo, Alan and Sudan, Madhu},
  booktitle = {Approximation, Randomization, and Combinatorial Optimization. Algorithms and Techniques (APPROX/RANDOM)},
  title     = {{List Decoding Group Homomorphisms Between Supersolvable Groups}},
  year      = {2014},
  doi       = {10.4230/LIPIcs.APPROX-RANDOM.2014.737},
  isbn      = {978-3-939897-74-3},
  issn      = {1868-8969},
  urn       = {urn:nbn:de:0030-drops-47359},
}

@article{Kiwi03,
title = {Algebraic testing and weight distributions of codes},
journal = {Theoretical Computer Science},
volume = {299},
number = {1},
pages = {81--106},
year = {2003},
issn = {0304-3975},
doi = {10.1016/S0304-3975(02)00816-2},
author = {M. Kiwi},
}

@InProceedings{DSW06,
  author    = {Dinur, Irit and Sudan, Madhu and Wigderson, Avi},
  booktitle = {Approximation, Randomization, and Combinatorial Optimization. Algorithms and Techniques},
  title     = {Robust Local Testability of Tensor Products of {LDPC} Codes},
  year      = {2006},
  pages     = {304--315},
  publisher = {Springer Berlin Heidelberg},
  doi       = {10.1007/11830924_29},
  isbn      = {978-3-540-38045-0},
}

@InProceedings{BCH+95,
  author    = {M. Bellare and D. Coppersmith and J. H{\aa}stad and M. Kiwi and M. Sudan},
  booktitle = focs95,
  year      = {1995},
  title     = {Linearity testing in characteristic two},
  pages     = {432--441},
}

@MISC {suji,
    TITLE = {Image of a fixed element under a random endomorphism in an Abelian group},
    AUTHOR = {Robin Chapman},
    HOWPUBLISHED = {MathOverflow},
    NOTE = {\url{https://mathoverflow.net/q/30378} (version: 2010-07-04)},
    year = {2010}
}

@Article{HW03,
  author       = {Johan H{\aa}stad and Avi Wigderson},
  year         = {2003},
  journal = {Random Structures \& Algorithms},
  title        = {Simple Analysis of Graph Tests for Linearity and {PCP}},
  number       = {2},
  pages        = {139--160},
  volume       = {22},
  doi = {10.1002/rsa.10068}
}

@InProceedings{DGKS08,
  author    = {Dinur, Irit and Grigorescu, Elena and Kopparty, Swastik and Sudan, Madhu},
  booktitle = stoc08,
  title     = {Decodability of group homomorphisms beyond the {Johnson} bound},
  year      = {2008},
  pages     = {275--284},
  doi       = {10.1145/1374376.1374418},
  isbn      = {9781605580470}}

@Article{BCLR07,
  author    = {Ben-Or, Michael and Coppersmith, Don and Luby, Mike and Rubinfeld, Ronitt},
  title     = {Non-{A}belian homomorphism testing, and distributions close to their self-convolutions},
  issn      = {1098-2418},
  number    = {1},
  pages     = {49--70},
  doi       = {10.1002/rsa.20182},
  volume    = {32},
  journal   = {Random Structures \& Algorithms},
  month     = aug,
  publisher = {Wiley},
  year      = {2007},
}

@InProceedings{BK21,
	title = {Optimal inapproximability of satisfiable k-{LIN} over non-abelian groups},
	author = {Bhangale, Amey and Khot, Subhash},
	year = {2021},
	doi = {10.1145/3406325.3451003},
	booktitle = stoc21
}

@misc{BFL03,
	title={Near representations of finite groups},
	author={Babai, L{\'a}szl{\'o} and Friedl, Katalin and Luk{\'a}cs, Andr{\'a}s},
	note={Manuscript},
	year={2003},
}

@article{MR15,
	title={Approximate representations, approximate homomorphisms, and low-dimensional embeddings of groups},
	author={Moore, Cristopher and Russell, Alexander},
	journal={SIAM Journal on Discrete Mathematics},
	volume={29},
	number={1},
	pages={182--197},
	year={2015},
	publisher={SIAM},
	doi = {10.1137/140958578},
	eprint = {1009.6230}
}

@article{OY16,
  title = {Testing properties of functions on finite groups},
  volume = {49},
  ISSN = {1098-2418},
  DOI = {10.1002/rsa.20639},
  number = {3},
  journal = {Random Structures \& Algorithms},
  publisher = {Wiley},
  author = {Oono,  Kenta and Yoshida,  Yuichi},
  year = {2016},
  month = feb,
  pages = {579--598},
  eprint ={1509.00930}
}

@article{GH17,
	title={Inverse and stability theorems for approximate representations of finite groups},
	author={Gowers, William Timothy and Hatami, Omid},
	journal={Sbornik: Mathematics},
	volume={208},
	number={12},
	pages={1784},
	year={2017},
	publisher={IOP Publishing},
	eprint = {1510.04085},
	doi = {10.1070/SM8872}
}

@article{Gow08,
  title = {Quasirandom {{Groups}}},
  author = {Gowers, W. T.},
  year = {2008},
  month = may,
  journal = {Combinatorics, Probability and Computing},
  volume = {17},
  number = {3},
  pages = {363--387},
  publisher = {{Cambridge University Press}},
  issn = {1469-2163, 0963-5483},
  doi = {10.1017/S0963548307008826},
  urldate = {2024-02-10},
  langid = {english},
  eprint = {0710.3877}
}

@Article{Sanders10,
  author       = {Tom Sanders},
  year         = {2012},
  journal = {Analysis \& PDE},
  title        = {On the {B}ogolyubov-{R}uzsa Lemma},
  volume = {5},
  number = {3},
  eprint          = {1011.0107},
  doi = {10.2140/apde.2012.5.627}
}

@InProceedings{Sam07,
  author    = {A. Samorodnitsky},
  booktitle = stoc07,
  title     = {Low-degree tests at large distances},
  year      = {2007},
}

@InProceedings{BLR90,
  author    = {M. Blum and M. Luby and R. Rubinfeld},
  booktitle = stoc90,
  year      = {1990},
  title     = {Self-Testing/Correcting with Applications to Numerical Problems},
  pages     = {73--83},
  doi = {10.1145/100216.100225}
}

@InProceedings{BSVW03,
  author    = {Eli {Ben-Sasson} and Madhu Sudan and Salil P. Vadhan and Avi Wigderson},
  booktitle = stoc03,
  year      = 2003,
  title     = {Randomness-efficient low degree tests and short {PCP}s via $\epsilon$-biased sets},
  pages     = {612-621},
  doi = {10.1145/780542.780631}
}
\end{document}